\documentclass[11pt,notitlepage,tightenlines,nofootinbib,superscriptaddress]{revtex4-2}

\usepackage[table]{xcolor}

\usepackage{newpxtext,newpxmath}

\usepackage[latin1]{inputenc}
\usepackage{amsthm}
\usepackage{amssymb}
\usepackage{amsmath}
\usepackage{bbold}
\usepackage{bbm}
\usepackage[pdftex, backref=page]{hyperref}
\usepackage{braket}
\usepackage{dsfont}
\usepackage{mathdots}
\usepackage{mathtools}
\usepackage{enumerate}
\usepackage[shortlabels]{enumitem}
\usepackage{csquotes}
\usepackage{stmaryrd}
\usepackage[cal=boondox]{mathalfa}
\usepackage{graphicx}
\usepackage{stackengine}
\usepackage{scalerel}
\usepackage{tensor}       %\tensor[_n]{\braket{j|\psi}}{_{1\ldots n}}
\usepackage{array}
\usepackage{makecell}
\newcolumntype{x}[1]{>{\centering\arraybackslash}p{#1}}
\usepackage{tikz}
\usepackage{pgfplots}
\usetikzlibrary{shapes.geometric, shapes.misc, positioning, arrows, arrows.meta, decorations.pathreplacing, decorations.pathmorphing, patterns, angles, quotes, calc}
\usepackage{booktabs}
\usepackage{xfrac}
\usepackage{siunitx}
\usepackage{centernot}
\usepackage{comment}
\usepackage{chngcntr}
\usepackage[justification=justified,format=plain]{caption}
\usepackage{subcaption}
\usepackage{ragged2e}

\newcommand{\ra}[1]{\renewcommand{\arraystretch}{#1}}

\newtheorem{thm}{Theorem}
\newtheorem*{thm*}{Theorem}
\newtheorem{prop}[thm]{Proposition}
\newtheorem*{prop*}{Proposition}
\newtheorem{lemma}[thm]{Lemma}
\newtheorem*{lemma*}{Lemma}
\newtheorem{cor}[thm]{Corollary}
\newtheorem*{cor*}{Corollary}

\newtheorem*{cj*}{Conjecture}

\newtheorem*{Def*}{Definition}

\newtheorem*{question*}{Question}

\newtheorem*{problem*}{Problem}

\makeatletter
\def\thmhead@plain#1#2#3{%
  \thmname{#1}\thmnumber{\@ifnotempty{#1}{ }\@upn{#2}}%
  \thmnote{ {\the\thm@notefont#3}}}
\let\thmhead\thmhead@plain
\makeatother

\theoremstyle{definition}
\newtheorem{rem}[thm]{Remark}

\makeatletter
\newcommand{\manualifempty}[3]{%
  \edef\@tempa{#1}%
  \ifx\@tempa\@empty
    #2% true branch
  \else
    #3% false branch
  \fi
}
\makeatother

\makeatletter
\newtheoremstyle{manualstyle}
  {3pt}{3pt}{\itshape}{}{\bfseries}{.}{ }{}

\theoremstyle{manualstyle}

\newtheorem{manualthminner}{Theorem}

\newtheorem{manualpropinner}{Proposition}

\newtheorem{manuallemmainner}{Lemma}

\newtheorem{manualcorinner}{Corollary}

\makeatother

\newcommand{\bb}{\begin{equation}\begin{aligned}\hspace{0pt}}
\newcommand{\bbb}{\begin{equation*}\begin{aligned}}
\newcommand{\ee}{\end{aligned}\end{equation}}
\newcommand{\eee}{\end{aligned}\end{equation*}}
\newcommand\floor[1]{\lfloor#1\rfloor}

\newcommand{\eqt}[1]{\stackrel{\mathclap{\scriptsize \mbox{#1}}}{=}}
\newcommand{\leqt}[1]{\stackrel{\mathclap{\scriptsize \mbox{#1}}}{\leq}}

\newcommand{\geqt}[1]{\stackrel{\mathclap{\scriptsize \mbox{#1}}}{\geq}}
\newcommand{\ketbra}[1]{\ket{#1}\!\!\bra{#1}}

\newcommand{\sumno}{\sum\nolimits}

\newcommand{\e}{\varepsilon}
\renewcommand{\epsilon}{\varepsilon}

\newcommand{\dd}{\mathrm{d}}

\newcommand{\id}{\mathds{1}}
\newcommand{\R}{\mathds{R}}
\newcommand{\N}{\mathds{N}}

\newcommand{\C}{\mathds{C}}
\newcommand{\E}{\mathds{E}}

\newcommand{\SEP}{\pazocal{S}}

\DeclareMathOperator{\Tr}{Tr}

\DeclareMathOperator{\co}{conv}
\DeclareMathOperator{\cone}{cone}

\DeclareMathAlphabet{\pazocal}{OMS}{zplm}{m}{n}

\DeclareMathOperator{\pr}{Pr}
\DeclareMathOperator{\supp}{supp}

\DeclareMathOperator{\relint}{relint}

\newcommand{\HH}{\pazocal{H}}

\newcommand{\MM}{\pazocal{M}}
\newcommand{\D}{\pazocal{D}}

\newcommand{\CC}{\pazocal{C}}

\newcommand{\XX}{\pazocal{X}}
\newcommand{\TT}{\pazocal{T}}
\newcommand{\PP}{\pazocal{P}}
\newcommand{\FF}{\pazocal{F}}

\newcommand{\lsmatrix}{\left(\begin{smallmatrix}}
\newcommand{\rsmatrix}{\end{smallmatrix}\right)}

\newcommand{\deff}[1]{\textbf{\emph{#1}}}
\newcommand{\rel}[3]{#1\big(#2\,\big\|\,#3\big)}
\newcommand{\Rel}[3]{#1\Big(#2\,\Big\|\,#3\Big)}

\stackMath
\newcommand\xxrightarrow[2][]{\mathrel{%
  \setbox2=\hbox{\stackon{\scriptstyle#1}{\scriptstyle#2}}%
  \stackunder[5pt]{%
    \xrightarrow{\makebox[\dimexpr\wd2\relax]{$\scriptstyle#2$}}%
  }{%
   \scriptstyle#1\,%
  }%
}}

\newcommand{\tends}[2]{\xxrightarrow[\! #2 \!]{\mathrm{#1}}}

\newcommand{\ctends}[3]{\xxrightarrow[\raisebox{#3}{$\scriptstyle #2$}]{\raisebox{-0.7pt}{$\scriptstyle \mathrm{#1}$}}}

\stackMath

\makeatletter
\newcommand*\rel@kern[1]{\kern#1\dimexpr\macc@kerna}
\newcommand*\widebar[1]{%
  \begingroup
  \def\mathaccent##1##2{%
    \rel@kern{0.8}%
    \overline{\rel@kern{-0.8}\macc@nucleus\rel@kern{0.2}}%
    \rel@kern{-0.2}%
  }%
  \macc@depth\@ne
  \let\math@bgroup\@empty \let\math@egroup\macc@set@skewchar
  \mathsurround\z@ \frozen@everymath{\mathgroup\macc@group\relax}%
  \macc@set@skewchar\relax
  \let\mathaccentV\macc@nested@a
  \macc@nested@a\relax111{#1}%
  \endgroup
}

\counterwithin*{equation}{part}
\counterwithin*{thm}{part}
\counterwithin*{figure}{part}

\tikzset{meter/.append style={draw, inner sep=10, rectangle, font=\vphantom{A}, minimum width=30, line width=.8, path picture={\draw[black] ([shift={(.1,.3)}]path picture bounding box.south west) to[bend left=50] ([shift={(-.1,.3)}]path picture bounding box.south east);\draw[black,-latex] ([shift={(0,.1)}]path picture bounding box.south) -- ([shift={(.3,-.1)}]path picture bounding box.north);}}}
\tikzset{roundnode/.append style={circle, draw=black, fill=gray!20, thick, minimum size=10mm}}
\tikzset{squarenode/.style={rectangle, draw=black, fill=none, thick, minimum size=10mm}}

\definecolor{Blues5seq1}{RGB}{239,243,255}
\definecolor{Blues5seq2}{RGB}{189,215,231}
\definecolor{Blues5seq3}{RGB}{107,174,214}
\definecolor{Blues5seq4}{RGB}{49,130,189}
\definecolor{Blues5seq5}{RGB}{8,81,156}

\definecolor{Greens5seq1}{RGB}{237,248,233}
\definecolor{Greens5seq2}{RGB}{186,228,179}
\definecolor{Greens5seq3}{RGB}{116,196,118}
\definecolor{Greens5seq4}{RGB}{49,163,84}
\definecolor{Greens5seq5}{RGB}{0,109,44}

\definecolor{Reds5seq1}{RGB}{254,229,217}
\definecolor{Reds5seq2}{RGB}{252,174,145}
\definecolor{Reds5seq3}{RGB}{251,106,74}
\definecolor{Reds5seq4}{RGB}{222,45,38}
\definecolor{Reds5seq5}{RGB}{165,15,21}

\allowdisplaybreaks

\usepackage[most,breakable]{tcolorbox}
\makeatletter

\def\boxed@gobegin#1{\def\@tempa{#1}\def\@tempb{orange}\ifx\@tempa\@tempb\begin{tcolorbox}[colback=red!15,colframe=orange!70,breakable,enhanced]\else\begin{tcolorbox}[colback=Blues5seq1,colframe=Blues5seq5,breakable,enhanced]\fi}
\def\boxed@gobegin@empty{\begin{tcolorbox}[colback=Blues5seq1,colframe=Blues5seq5,breakable,enhanced]}
\makeatother

\usepackage{adjustbox}
\usepackage{multirow}

\newcommand{\stein}{\mathrm{Stein}}

\newcommand{\iid}{\mathrm{iid}}
\newcommand{\av}{\mathrm{av}}
\newcommand{\all}{\mathds{ALL}}
\newcommand{\RR}{\pazocal{R}}

\renewcommand{\SS}{\pazocal{S}}
\renewcommand{\AA}{\pazocal{A}}
\newcommand{\BB}{\pazocal{B}}
\newcommand{\KK}{\pazocal{K}}
\newcommand{\stab}{\mathrm{STAB}}
\renewcommand{\D}{\mathcal{D}}

\renewcommand{\SEP}{\mathrm{SEP}}

\newcommand{\mmq}[2]{\pazocal{D}_{#1,#2}}
\newcommand{\mmv}[1]{\pazocal{D}_{\delta,#1}}

\renewcommand{\deff}[1]{\emph{#1}}

\newtheorem{newax}{Axiom}

\makeatletter

\renewcommand*{\p@subsection}{}

\renewcommand*{\p@subsubsection}{}
\makeatother

\DeclareMathOperator*{\minst}{min^\star\hspace{-4pt}}

\begin{document}

\title{Generalised quantum Sanov theorem revisited
%Doubly composite classical hypothesis testing
}

\author{Ludovico Lami}
\email{ludovico.lami@gmail.com}
\affiliation{Scuola Normale Superiore, Piazza dei Cavalieri 7, 56126 Pisa, Italy}

\begin{abstract}
Given two families of quantum states $\pazocal{A}$ and $\pazocal{B}$, called the null and the alternative hypotheses, quantum hypothesis testing is the task of determining whether an unknown quantum state belongs to $\pazocal{A}$ or $\pazocal{B}$. Mistaking $\pazocal{A}$ for $\pazocal{B}$ is a type I error, and vice versa for the type II error. In quantum Shannon theory, a fundamental role is played by the Stein exponent, i.e.\ the asymptotic rate of decay of the type II error probability for a given threshold on the type I error probability. Stein exponents have been thoroughly investigated --- and, sometimes, calculated. However, most currently available solutions apply to settings where the hypotheses simple (i.e.\ composed of a single state), or else the families $\pazocal{A}$ and $\pazocal{B}$ need to satisfy stringent constraints that exclude physically important sets of states, such as separable states or stabiliser states. In this work, we establish a general formula for the Stein exponent where both hypotheses are allowed to be composite: the alternative hypothesis $\pazocal{B}$ is assumed to be either composite i.i.d.\ or arbitrarily varying, with components taken from a known base set, while the null hypothesis $\pazocal{A}$ is fully general, and required to satisfy only weak compatibility assumptions that are met in most physically relevant cases --- for instance, by the sets of separable or stabiliser states. Our result extends and subsumes the findings of \href{https://link.springer.com/article/10.1007/s00220-021-04133-8}{[BBH, CMP 385:55, 2021]} (that we also simplify), as well as the `generalised quantum Sanov theorem' of \href{https://arxiv.org/abs/2408.07067}{[LBR, arXiv:2408.07067]}. The proof relies on a careful quantum-to-classical reduction via measurements, followed by an application of the results on classical Stein exponents obtained in~[Lami, arXiv:today]. We also devise new purely quantum techniques to analyse the resulting asymptotic expressions.
\end{abstract}

\maketitle

%\tableofcontents

\section{Introduction}

\subsection{Background and motivation} \label{subsec_background}

Hypothesis testing is an essential primitive of classical as well as quantum information theory, underpinning much of its technical machinery. It is deeply connected with coding theory, and, more generally, with quantum resource distillation, namely, the general family of tasks where one needs to refine quantum resources (either states or channels). The key reason that underlies this connection is simple enough: any quantum resource distillation protocol can be used to test whether the resource is there in the first place --- if it is, then the protocol will refine it so that it is easily detected. In several situations this reasoning can be reversed, so that solving a quantum hypothesis testing problem yields a solution to the corresponding quantum resource distillation problem. 

One of the cornerstone results of quantum information theory is Hiai and Petz's quantum Stein's lemma~\cite{Hiai1991, Ogawa2000}, which, extending the classical Chernoff--Stein lemma~\cite{stein_unpublished, chernoff_1956}, determines the asymptotic rate of decay of the error probability in asymmetric hypothesis testing between two i.i.d.\ hypotheses. Historically, Hiai and Petz's result has played a key role in the theory, as it has decisively contributed to the identification of Umegaki's relative entropy~\cite{Umegaki1962} as the operational analogue of the classical Kullback--Leibler divergence~\cite{Kullback-Leibler}. The impact this result has had throughout quantum Shannon theory is hard to overestimate~\cite{Ogawa2002, Hayashi2007, HAYASHI, TOMAMICHEL}.

For many applications, however, we are not content with the solution of this simple setting where both hypotheses are \emph{simple} (i.e.\ composed of one single state) and furthermore i.i.d.; we want to consider instead more general classes of hypotheses that are both \emph{composite} (i.e.\ represented by sets of states) and \emph{genuinely correlated}. The general problem can thus be phrased as follows: given a quantum state $\rho_n\in \D\big(\HH^{\otimes n}\big)$ of a quantum system composed of $n$ copies of an elementary system with Hilbert space $\HH$, we have to decide between two options:
\begin{itemize}
\item[$\mathrm{H}_0$.] Null hypothesis: $\rho_n \in \AA_n$;
\item[$\mathrm{H}_1$.] Alternative hypothesis: $\rho_n \in \BB_n$;
\end{itemize}
Here, $\AA_n,\BB_n\subseteq \D\big(\HH^{\otimes n}\big)$ are two known families of states, which we will collect in the sequences $\AA = (\AA_n)_n$ and $\BB = (\BB_n)_n$. With a slight abuse of language, we refer also to $\AA$ and $\BB$ as the hypotheses. We say that one of the two hypotheses, say $\AA$, is composite if the sets $\AA_n$ do not contain a single state. We call it instead genuinely correlated if the extreme points of $\AA_n$ are not all tensor products across the copies, i.e.\ of the form $\rho_1\otimes \ldots \otimes \rho_n$. Two notable examples of composite but not genuinely correlated classes of hypotheses stand out: given some set of states $\FF_1\subseteq \D(\HH)$ over a finite-dimensional Hilbert space $\HH$, we refer to the family
\bb
\FF = \FF_1^\iid &\coloneqq \scaleobj{1.15}{\big(} \FF_1^{\otimes n,\,\iid}\scaleobj{1.15}{\big)}_n\, \qquad \FF_1^{\otimes n,\,\iid} \coloneqq \big\{\rho^{\otimes n}:\, \rho\in \FF_1\big\}
\label{q_F_n_iid} 
\ee
as the associated \emph{composite i.i.d.\ hypothesis}, and to the family
\bb
\FF = \FF_1^\av &\coloneqq \scaleobj{1.15}{\big(} \FF_1^{\otimes n,\,\av}\scaleobj{1.15}{\big)}_n\, \qquad \FF_1^{\otimes n,\, \av} \coloneqq \big\{\rho_1\!\otimes\! \ldots \!\otimes \rho_n\!:\ \rho_1,\ldots, \rho_n\in \FF_1\big\}
\label{q_F_n_av}
\ee
as the associated \emph{composite arbitrarily varying hypothesis}. 

Guessing which one of the two hypotheses holds amounts to making a measurement, modelled by the binary positive operator-valued measure (POVM) $(E_n,\id-E_n)$, where $E_n$ corresponds to guessing $\mathrm{H}_0$, and $\id-E_n$ to the complementary event of guessing $\mathrm{H}_1$. Here, $E_n$ is an a priori arbitrary operator on $\HH^{\otimes n}$ obeying the two-fold inequality $0\leq E_n\leq \id$, where, for two operators $X,Y$ on the same finite-dimensional Hilbert space, we write $X\leq Y$ if $Y-X$ is positive semi-definite. A \deff{type I error} is defined as guessing $\mathrm{H}_1$ when $\mathrm{H}_0$ holds, and vice versa for the \deff{type II error}. The worst-case error probabilities of making these two errors are given by
\bb
\alpha_n(E_n) \coloneqq \sup_{\rho_n\in \AA_n} \Tr \left[ \rho_n (\id-E_n)\right] ,\qquad \beta_n(E_n) \coloneqq \sup_{\rho_n\in \BB_n} \Tr\left[\rho_n E_n\right] ,
\label{error_probabilities}
\ee
respectively, where we kept the dependence on the two sets $\AA_n$ and $\BB_n$ implicit. The trade-off between the two error probabilities can then be represented by the function
\bb
\beta_\e(\AA_n\|\BB_n) \coloneqq \inf\left\{ \beta_n(E_n):\ \ 0\leq E_n\leq \id,\ \alpha_n(E_n)\leq \e \right\} .
\label{beta_e_level_n}
\ee
In analogy with the simple i.i.d.\ setting, we capture the asymptotic behaviour of this function by defining the associated \deff{Stein exponent} as
\bb
\stein(\AA\|\BB) \coloneqq \lim_{\e\to 0^+} \liminf_{n\to\infty} \left\{ -\frac1n\log \beta_\e(\AA_n\|\BB_n)\right\} .
\label{Stein}
\ee
The fundamental quantum Stein's lemma allows us to calculate this limit in the case where $\AA = \{\rho\}^\iid = \big( \{\rho^{\otimes n}\} \big)_n$ and $\BB = \{\sigma\}^\iid = \big( \{\sigma^{\otimes n}\} \big)_n$ are two simple i.i.d.\ hypotheses. It can be stated as~\cite{Hiai1991}
\bb
\stein(\rho\|\sigma) = D(\rho\|\sigma) \coloneqq \Tr \left[\rho \left(\log \rho - \log \sigma\right) \right] ,
\label{Hiai_Petz}
\ee
where we wrote more compactly $\stein(\rho\|\sigma)$ instead of $\rel{\stein}{\{\rho\}^\iid}{\{\sigma\}^\iid}$.
Numerous extensions of Hiai and Petz's quantum Stein's lemma have been proposed, to calculate the Stein exponent in composite settings of ever increasing complexity and generality. Let us survey some of these contributions (see also Table~\ref{results_table}).

The case of a composite i.i.d.\ null hypothesis was tackled in~\cite[Theorem~2.2]{Bjelakovic2005}, while that of an arbitrarily varying null hypothesis was addressed in~\cite[Theorem~1]{Noetzel2014}. In both of these works, however, the alternative hypothesis is assumed to be simple and i.i.d.; a generalisation that covers the case where \emph{both} hypotheses --- null and alternative --- are composite (but still i.i.d.)\ was put forth in~\cite[Theorem~1.1]{berta_composite}. See also~\cite{Mosonyi2022} for related results.

So far, we have only described settings that involve non-genuinely correlated hypotheses. A major step forward was the realisations that more general, genuinely correlated hypotheses can also be analysed. A paradigmatic example is the setting of \emph{entanglement testing}~\cite{gap-comment}, where the underlying elementary Hilbert space $\HH = \HH_A\otimes \HH_B$ is bipartite, and one of the two hypotheses is of the form $\FF = \big(\SEP_n\big)_n$, where $\SEP_n \coloneqq \SEP_{A^n:B^n}$ comprises all states that are \emph{separable}~\cite{Werner}, i.e.\ un-entangled, across the cut $A^n:B^n$. Importantly, this is an example of a genuinely correlated hypothesis, because there \emph{can} be entanglement across the cuts $A_1B_1:A_2B_2:\ldots :A_nB_n$. Entanglement testing is just an instance of a more general class of tasks, collectively called \emph{resource testing}, in which one looks at other interesting quantum resources, such as nonstabiliser-ness (a.k.a.\ `quantum magic')~\cite{Veitch2014}. The importance of resource testing is not only --- rather obviously --- for practical applications such as device certification, in which we want to ascertain whether a device truly outputs resourceful states; more fundamentally, it is profoundly connected with quantum resource manipulation~\cite{Vedral1997, HAYASHI, BrandaoPlenio1, BrandaoPlenio2}. 

Cornerstones of this connection are two fundamental results known as the `generalised quantum Stein's lemma', whose proof unravelled into a saga that concluded only recently~\cite{Brandao2010, gap, gap-comment, Hayashi-Stein, GQSL}, and the `generalised quantum Sanov theorem'~\cite{generalised-Sanov} (see also~\cite{Hayashi-Sanov-1, Hayashi-Sanov-2}). These two results are, in a precise sense, complementary to each other: the former gives a closed-form expression for the Stein exponent of a hypothesis testing task in which the null hypothesis is a simple and i.i.d., and the alternative hypothesis comprises all resourceless (also called `free') states; the same is true of the latter result upon swapping the two hypotheses. 

The final step in this march towards ever greater generality is to consider the setting in which \emph{both} hypotheses are composite and genuinely correlated. Some progress has been made in this direction too, but previous results are typically subjected to significant limitations, in that they either consider restricted sets of measurements that are designed to effectively `tame' the correlations among the systems~\cite{Piani2009, brandao_adversarial}, or they impose strong restrictions on the families of hypotheses that can be treated~\cite{Fang2025}, excluding, for example, separable states as well as stabiliser (i.e.\ non-magic) states.

\subsection{Our contribution}

In this work, we take a further step towards the construction of a more flexible framework, one that is capable of accommodating a broader class of composite and genuinely correlated hypotheses, in particular encompassing strongly correlated families such as separable and stabiliser states. Our main result is Theorem~\ref{stronger_genq_Sanov_thm}, which gives a general formula for the Stein exponent in the case where:
\begin{itemize}
\item the null hypothesis is a general set of states obeying only weak assumptions, satisfied by most sets of states of interest, including the sets of separable and stabiliser states; and
\item the alternative hypothesis is either composite i.i.d.\ (as defined in~\eqref{q_F_n_iid}) or arbitrarily varying (as defined in~\eqref{q_F_n_av}), with components taken from a base set of quantum states $\BB_1$. We denote these two hypotheses as $\BB_1^\iid$ and $\BB_1^\av$, respectively.
\end{itemize}
Thus, Theorem~\ref{stronger_genq_Sanov_thm} extends the `generalised quantum Sanov theorem' of~\cite{generalised-Sanov}, which only covered the case of a simple i.i.d.\ alternative hypothesis. See Table~\ref{results_table} for a schematic representation of the difference between these two results. Interestingly, among other things Theorem~\ref{stronger_genq_Sanov_thm} shows that the Stein exponent corresponding to an arbitrarily varying alternative hypothesis $\BB_1^\av$ \emph{coincides} with that obtained in the case of a composite \emph{i.i.d.}\ alternative hypothesis of the form $\co(\BB_1)^\iid$, where $\co(\BB_1)$ denotes the convex hull of $\BB_1$. In other words, these two scenarios are entirely equivalent up to replacing the base set $\BB_1$ with its convex hull.

As immediate consequences of Theorem~\ref{stronger_genq_Sanov_thm}, we deduce two characterisations of the false negative error exponents for both entanglement testing (Corollary~\ref{entanglement_testing_cor}) and magic testing (Corollary~\ref{magic_cor}). In Theorem~\ref{q_both_composite_iid_or_av_thm}, we also show how to apply Theorem~\ref{stronger_genq_Sanov_thm} to refine prior results of~\cite{berta_composite}, extending them to the case of arbitrarily varying hypotheses and simplifying the resulting formulas.

On the technical side, our proofs are based on a few different ingredients. First, a careful quantum-to-classical reduction obtained by measuring, where the quantum measurement to be performed is chosen via minimax. The second step is to apply the solution of the Stein exponent in the composite and genuinely correlated \emph{classical} setting of~\cite[Theorems~2 and~4]{doubly-comp-classical}. (As explained in~\cite{doubly-comp-classical}, this solution in turn relies on the `symbol-by-symbol' blurring technique, an extension of the method of blurring introduced in~\cite{GQSL}.) Via a double blocking technique, we thus arrive at a first expression for the quantum Stein exponent. In the last steps of the proof, we bring to bear new purely quantum techniques to simplify it further, showing, in particular, that the Stein exponents corresponding to the two alternative hypotheses $\BB_1^\av$ and $\co(\BB_1)^\iid$ in Theorem~\ref{stronger_genq_Sanov_thm} coincide. The crucial step in this direction is made possible by Proposition~\ref{av_to_iid_reduction_reg_relent_prop}. Finally, in the setting of~\cite{berta_composite} we utilise a variant of the Alicki--Fannes--Winter method from~\cite{Alicki-Fannes, tightuniform, Shirokov-review} to obtain a formula for the Stein exponent (Theorem~\ref{q_both_composite_iid_or_av_thm}) that is both simpler and more general than that in~\cite{berta_composite}.

\begin{table}[h] \centering
\ra{1.5}
\begin{tabular}{*7c} \toprule
\multirow{2}{*}{\large Result} & \multicolumn{3}{c}{\large Null hypothesis} & \multicolumn{3}{c}{\large Alternative hypothesis} \\
%\multirow{2}{*}{$\begin{array}{c} \\[-2ex] \text{Includes SEP} \\[-1.6ex] \text{and STAB} \end{array}$} \\
& Composite & $\begin{array}{c} \text{Genuinely} \\[-1.6ex] \text{correlated} \end{array}$ & $\begin{array}{c} \\[-2ex] \text{Includes} \\[-1.6ex] \text{SEP\! \&\! STAB} \end{array}$ & Composite & $\begin{array}{c} \text{Genuinely} \\[-1.6ex] \text{correlated} \end{array}$ & $\begin{array}{c} \\[-2ex] \text{Includes} \\[-1.6ex] \text{SEP\! \&\! STAB} \end{array}$ \\
\midrule
Q.\ Stein's lemma~\cite{Hiai1991} & \cellcolor{Reds5seq3} N & \cellcolor{Reds5seq3} N & \cellcolor{Reds5seq3} N & \cellcolor{Reds5seq3} N & \cellcolor{Reds5seq3} N & \cellcolor{Reds5seq3} N \\
Q.\ Sanov theorem~\cite{Bjelakovic2005, Noetzel2014} & \cellcolor{Greens5seq3} Y & \cellcolor{Reds5seq3} N & \cellcolor{Reds5seq3} N & \cellcolor{Reds5seq3} N & \cellcolor{Reds5seq3} N & \cellcolor{Reds5seq3} N \\
BBH's extension~\cite{berta_composite} & \cellcolor{Greens5seq3} Y & \cellcolor{Reds5seq3} N & \cellcolor{Reds5seq3} N & \cellcolor{Greens5seq3} Y & \cellcolor{Reds5seq3} N & \cellcolor{Reds5seq3} N \\
%BBGLPRT's results~\cite{gap} & \cellcolor{Reds5seq3} N & \cellcolor{Reds5seq3} N & \cellcolor{Reds5seq3} N & \cellcolor{Greens5seq3} Y & \cellcolor{Reds5seq3} N & \cellcolor{Reds5seq3} N \\
$\begin{array}{c} \text{Generalised q.\ Stein's} \\[-1.5ex] \text{lemma (version of~\cite{Hayashi-Stein})} \end{array}$ & \cellcolor{Reds5seq3} N & \cellcolor{Reds5seq3} N & \cellcolor{Reds5seq3} N & \cellcolor{Greens5seq3} Y & \cellcolor{Greens5seq3} Y & \cellcolor{Greens5seq3} Y \\
$\begin{array}{c} \text{Generalised q.\ Stein's} \\[-1.5ex] \text{lemma (version of~\cite{GQSL})} \end{array}$ & \cellcolor{yellow!90!black} Y & \cellcolor{yellow!90!black} Y & \cellcolor{Reds5seq3} N & \cellcolor{Greens5seq3} Y & \cellcolor{Greens5seq3} Y & \cellcolor{Greens5seq3} Y \\
$\begin{array}{c} \text{Generalised q.\ Sanov} \\[-1.5ex] \text{theorem~\cite{generalised-Sanov}} \end{array}$ & \cellcolor{Greens5seq3} Y & \cellcolor{Greens5seq3} Y & \cellcolor{Greens5seq3} Y & \cellcolor{Reds5seq3} N & \cellcolor{Reds5seq3} N & \cellcolor{Reds5seq3} N \\
$\begin{array}{c} \text{\emph{Another} generalised q.} \\[-1.5ex] \text{Stein's lemma by FFF~\cite{Fang2025}} \end{array}$ & \cellcolor{Greens5seq3} Y & \cellcolor{Greens5seq3} Y & \cellcolor{Reds5seq3} N & \cellcolor{Greens5seq3} Y & \cellcolor{Greens5seq3} Y & \cellcolor{Reds5seq3} N \\
This work (Theorem~\ref{stronger_genq_Sanov_thm}) & \cellcolor{Greens5seq3} Y & \cellcolor{Greens5seq3} Y & \cellcolor{Greens5seq3} Y & \cellcolor{Greens5seq3} Y & \cellcolor{Reds5seq3} N & \cellcolor{Reds5seq3} N  \\
\bottomrule
\end{tabular}
\vspace{-1ex}
\captionof{table}[Prior results in quantum hypothesis testing]{ \justifying
Some representative results on Stein exponents in quantum hypothesis testing, classified according to whether the null and alternative hypotheses under consideration may be composite or genuinely correlated; in this latter case, we also indicated whether they are broad enough to include the strongly correlated families of separable states (SEP) and stabiliser states (STAB). A hypothesis, specified by a set of quantum states, is termed composite if the set contains more than one state, and genuinely correlated if not all of its extreme points are tensor product states across the copies. Inclusion of the feature is indicated by a green cell with `Y', exclusion by a red cell with `N'. The special case of~\cite[Theorem~32]{GQSL} has two cells coloured in yellow: although the null hypothesis there is formally both composite and genuinely correlated (it includes all `almost power states' along a fixed state $\rho$ with a constant number of defects), in practice this corresponds to an almost-i.i.d.\ scenario and is therefore not composite and genuinely correlated in the same spirit as the other cases. 

We restrict attention here to the ultimate limits of quantum hypothesis testing, corresponding to settings where arbitrary quantum measurements on the systems are allowed; scenarios with restricted measurement sets, studied in~\cite{brandao_adversarial}, are not included.
}
\label{results_table}
\end{table}

\section{Main results} \label{sec_main_results}

This section is devoted to the presentation of our main result (Theorem~\ref{stronger_genq_Sanov_thm} below) and some of its most notable consequences, such as Corollaries~\ref{entanglement_testing_cor} and~\ref{q_iid_vs_both_cor}. All proofs are deferred to Section~\ref{sec_proofs}. We start by expounding the assumptions underpinning our framework, which, to some extent, mimic the axioms employed in~\cite{doubly-comp-classical} in the classical case. An important role in our theory is played by the following special type of quantum channel. Given a Hilbert space $\HH$, some state $\tau\in \D(\HH)$, and some $\delta\in [0,1]$, the associated \deff{depolarising channel} $\mmv{\tau}$ is the super-operator $\mmv{\tau}$ that acts on the space of Hermitian operators on $\HH$ as
\bb
\mmv{\tau}(X) \coloneqq (1-\delta) X + \delta \tau\, .
\label{q_depolarising}
\ee

Now, for some Hilbert space $\HH$, consider a sequence $\FF = (\FF_n)_n$ of hypotheses $\FF_n\subseteq \D\big(\HH^{\otimes n}\big)$. We will now present the main compatibility assumptions that we will require on the null hypothesis to state our main result. The following can be thought of as a quantum version of~\cite[Axiom~I]{doubly-comp-classical}; however, it is strictly stronger, as~\cite[Axiom~I]{doubly-comp-classical} can be obtained from Axiom~\ref{q_ax_depolarising} below by restricting to the case where $k=1$:

{\renewcommand{\thenewax}{Q.\Roman{newax}}
\begin{newax} \label{q_ax_depolarising}
There exists some $\tau\in \FF_1$ such that, for all $m,k\in \N^+$ and all $\rho_{mk} \in \FF_{mk}$:
\vspace{-1.2ex}
\begin{enumerate}[(a), itemsep=0.5ex] 
\item $\supp(\rho_{mk}) \subseteq \supp(\tau)^{\otimes mk}$; and 
\item $\mmv{\tau^{\otimes k}}^{\otimes m}(\rho_{mk})\in \FF_{mk}$ for all $\delta\in [0,1]$, where $\mmv{\tau}$ is as in~\eqref{q_depolarising}. 
\end{enumerate}
\vspace{-1ex}
We denote as $c > 0$ a lower bound on the smallest non-zero eigenvalue of $\tau$, i.e.\ a constant with the property that $\min_{i:\ \lambda_i(\tau) > 0} \lambda_i(\tau) \geq c > 0$, where $\lambda_i(\tau)$ denotes the $i^\text{th}$ eigenvalue of $\tau$.
\end{newax}

We now introduce an analogously slightly stronger version of~\cite[Axiom~II]{doubly-comp-classical}:

\begin{newax} \label{q_ax_tensor_powers}
$(\FF_n)_n$ is closed under tensor powers, in the sense that $\rho_k^{\otimes m} \in \FF_{mk}$ for all $m,k\in \N^+$ and all $\rho_k\in \FF_k$.
\end{newax}

While stronger than~\cite[Axiom~II]{doubly-comp-classical}, which is again obtained by setting $k=1$, the above Axiom~\ref{q_ax_tensor_powers} is however strictly weaker than the tout court closure under tensor products required in~\cite[Property~4, p.~5]{Brandao2010} as well as in~\cite[Axiom~State2]{Hayashi-Stein} and in~\cite[Axiom~4, p.~6]{GQSL}. An important class of examples that does not satisfy these latter assumptions but does satisfy Axiom~\ref{q_ax_tensor_powers} is constituted by composite i.i.d.\ hypotheses (see~\eqref{q_F_n_iid}).

The following assumption, which we will require on the null hypothesis, is entirely standard:

\begin{newax} \label{q_ax_closed_permutations}
Each $\FF_n$ is closed under permutations: if $\rho_n \in \FF_n$ and $\pi\in S_n$ denotes an arbitrary permutation of a set of $n$ elements, then also $U_\pi^{\vphantom{\dag}} \rho_n U_\pi^\dag \in \FF_n$, where $U_\pi$ is the unitary that acts on $\HH^{\otimes n}$ by permuting the tensor factors. 
\end{newax}

Using a straightforward quantum analogue of the reasoning in~\cite[Lemma~26]{doubly-comp-classical}, one can show that Axioms~\ref{q_ax_depolarising}--\ref{q_ax_closed_permutations} are implied by --- and hence strictly weaker than --- the original Brand\~ao--Plenio axioms~\cite[Properties~1--5, pp.~4--5]{Brandao2010}. 
Rather surprisingly, however, it was demonstrated in~\cite[Appendix~E.2]{generalised-Sanov} that even the latter, when imposed on the null hypothesis, do not suffice to determine the Stein exponent, not even when the alternative hypothesis is simple and i.i.d. 
An additional assumption of a somewhat different nature must therefore be introduced. 
The following axiom, inspired by pioneering work of Piani~\cite{Piani2009}, replaces the classical~\cite[Axioms~IV--V]{doubly-comp-classical}:

\begin{newax} \label{q_ax_filtering}
For all $k\in \N^+$ there exists a neighbourhood $\pazocal{V}$ of $\id^{\otimes k}$ (the identity operator on $\HH^{\otimes k}$) such that, for all $F_k \in \pazocal{V}$, the map $X \mapsto \Tr[X F_k]$ is `completely $\cone(\FF)$-preserving'. This means that, for all $n\in \N^+$ and all states $\rho_{n+k}\in \FF_{n+k}$, we have 
\bb
\Tr_{n+1,\ldots,n+k} \left[ \rho_{n+k} \big(\id^{\otimes n} \otimes F_k \big)\right] \in \cone(\FF_n) \coloneqq \{\lambda \omega_n:\, \lambda\geq 0,\, \omega_n\in \FF_n\}\, ,
\label{complete_cone_preservation}
\ee
where $\Tr_{n+1,\ldots,n+k}$ denotes the partial trace over the last $k$ sub-systems.
\end{newax}

We are now ready to state our main result, whose name might admittedly benefit from a little more originality. 

\begin{thm}[(Generalised quantum Sanov theorem revisited)] \label{stronger_genq_Sanov_thm}
Let $\HH$ be a finite-dimensional Hilbert space, $\AA = (\AA_n)_n$ a sequence of sets $\AA_n\subseteq \D\big(\HH^{\otimes n}\big)$, and $\BB_1\subseteq \D(\HH)$ a non-empty and topologically closed set of states on $\HH$. Assume that all sets $\AA_n$ are topologically closed and convex, and that $\AA$ satisfies Axioms~\ref{q_ax_depolarising}--\ref{q_ax_filtering}. Then, using the notation in~\eqref{q_F_n_iid}--\eqref{q_F_n_av}, the Stein exponent defined by~\eqref{Stein} is given by
\bb
\rel{\stein}{\AA}{\BB_1^{\mathrm{iid}}} = 
%\rel{D^\infty}{\AA}{\co\big(\BB_1^{\mathrm{iid}}\big)} = 
%\lim_{n\to\infty} \frac1n\, \inf_{\sigma_n\in \co(\BB_1^{\otimes n,\, \mathrm{iid}})} D(\rho^{\otimes n} \| \sigma_n)\, ,
\minst_{\mu\in \PP(\BB_1)}\, \lim_{n\to\infty} \frac1n \inf_{\rho_n\in \AA_n} D\Big(\rho_n\, \Big\|\, \scaleobj{.8}{\int_{\BB_1}} \dd\mu(\sigma_1)\ \sigma_1^{\otimes n}\Big)\, ,
\label{stronger_genq_Sanov_iid}
\ee
where $\PP(\BB_1)$ denotes the set of all probability measures\footnote{That is, the set of all non-negative regular Borel measures $\mu$ on $\BB_1$ such that $\mu(\BB_1)=1$.} on the compact set $\BB_1$, and $\minst_\mu$ indicates that we restrict the minimisation to those $\mu$ such that the inner limit in $n$ exists (such a set is non-empty).
%, and the notation $\smash{\liminfsup{}}$ indicates that we can equivalently take either $\liminf$ or $\limsup$, and the result is the same. 
Similarly, 
\bb
\rel{\stein}{\AA}{\BB_1^{\mathrm{av}}} &= \rel{\stein}{\AA}{\co(\BB_1)^{\mathrm{av}}} \\
%&= \rel{\stein}{\AA}{\co(\BB_1)^{\mathrm{iid}}} \\
&= \rel{\stein}{\AA}{\co(\BB_1)^{\mathrm{iid}}} \\
%&= \rel{D^\infty}{\AA}{\co(\BB_1)^{\mathrm{av}}} = \rel{D^\infty}{\AA}{\co(\BB_1)^{\mathrm{iid}}} \\
&= \minst_{\mu\in \PP(\co(\BB_1))}\, \lim_{n\to\infty} \frac1n \inf_{\rho_n\in \AA_n} D\Big(\rho_n\, \Big\|\, \scaleobj{.8}{\int_{\co(\BB_1)}} \dd\mu(\sigma_1)\, \sigma_1^{\otimes n}\Big)\, .
\label{stronger_genq_Sanov_av}
\ee
\end{thm}

The proof is reported in Section~\ref{subsec_proof_stronger_genq_Sanov}. The above Theorem~\ref{stronger_genq_Sanov_thm} is a generalisation of~\cite[Theorem~14, Eq.~(D3)]{generalised-Sanov}, for it covers the case where the alternative hypothesis is composite i.i.d.\ or arbitrarily varying, rather than simple and i.i.d, but it also presents a significant drawback, in that it features a regularised expression for the Stein exponent instead of a single-letter one. This seems unavoidable, in the sense that the regularisation cannot be removed in an obvious way, not even when $\AA$ is simple and i.i.d.\ and $\BB_1$ contains two distinct states only~\cite[Theorem~IV.3]{Mosonyi2022}. In the case where $\BB_1$ is composed of a single state, instead, we can recover~\cite[Eq.~(D3)]{generalised-Sanov} from Theorem~\ref{stronger_genq_Sanov_thm} by using the additivity of the reverse relative entropy of resource~\cite[Eq.~(13)]{generalised-Sanov}.

The assumptions on $\AA$ posited in Theorem~\ref{stronger_genq_Sanov_thm} are satisfied by most physical interesting sets of states. These include, for example, the cases where $\AA = \SEP = \big( \SEP_n\big)_n$ is the sequence of sets of separable (i.e.\ un-entangled) states $\SEP_n \coloneqq \SEP_{A^n:B^n}$ on $n$ copies of a finite-dimensional bipartite quantum system $AB$, or where $\AA = \stab$ is the set of stabiliser states on an $N$-qubit system. Below we report the corollary we obtain in the former case, and we refer instead to the analogous Corollary~\ref{magic_cor} for the latter. Proofs can be found in Section~\ref{subsec_proofs_corollaries}.

\begin{cor} \label{entanglement_testing_cor}
Let $\HH_{AB}$ be a finite-dimensional bipartite Hilbert space. For some non-empty closed set $\FF_1\subseteq \D(\HH_{AB})$, the Stein exponents\footnote{These are also called the Sanov exponents in~\cite{generalised-Sanov}, to highlight the fact that the composite and genuinely correlated hypothesis is the null hypothesis.} of the entanglement testing tasks with null hypothesis given by the set of separable states and composite i.i.d.\ or arbitrarily varying alternative hypothesis, with base set $\FF_1$, can be expressed as
\bb
%\rel{\sanov}{\FF^{\mathrm{iid}}}{} = 
\rel{\stein}{\SEP}{\FF_1^{\mathrm{iid}}} &= 
%\rel{D^\infty}{\SEP}{\co\big(\FF_1^{\mathrm{iid}}\big)} = 
\minst_{\mu\in \PP(\FF_1)}\, \lim_{n\to\infty} \frac1n \inf_{\sigma_{A^n:B^n} \in\, \SEP_{A^n:B^n}} \!\!D\Big(\sigma_{A^n:B^n} \,\Big\|\, \scaleobj{.8}{\int_{\FF_1}} \!\!\dd\mu(\rho)\, \rho_{AB}^{\otimes n}\Big)
\label{entanglement_testing_iid}
\ee
and
\begin{align}
%\rel{\sanov}{\FF^{\mathrm{av}}}{} = 
\rel{\stein}{\SEP}{\FF_1^{\mathrm{av}}} &= \rel{\stein}{\SEP}{\co(\FF_1)^\mathrm{av}} \nonumber \\
&= \rel{\stein}{\SEP}{\co(\FF_1)^\mathrm{iid}} \label{entanglement_testing_av} \\
&= \minst_{\mu\in \PP(\co(\FF_1))}\, \lim_{n\to\infty} \frac1n \inf_{\sigma_{A^n:B^n} \in\, \SEP_{A^n:B^n}} D\Big(\sigma_{A^n:B^n} \,\Big\|\, \scaleobj{.8}{\int_{\co(\FF_1)}} \dd\mu(\rho)\ \rho_{AB}^{\otimes n}\Big)\, , \nonumber
\end{align}
respectively, where $\PP(\CC)$ denotes the set of probability measures on the compact set $\CC$, and $\minst_\mu$ indicates that we restrict the minimisation to those $\mu$ such that the inner limit in $n$ exists (such a set is non-empty).
\end{cor}

It is also possible to employ Theorem~\ref{stronger_genq_Sanov_thm} to refine and extend~\cite[Theorem~1.1]{berta_composite}, which deals with the case where both hypotheses are composite i.i.d. The following is however not a simple consequence of Theorem~\ref{stronger_genq_Sanov_thm}, as one can see by noting that the resulting formulas for the Stein exponents are a bit simpler than those in~\eqref{stronger_genq_Sanov_iid} and~\eqref{stronger_genq_Sanov_av}:

\begin{thm} \label{q_both_composite_iid_or_av_thm}
Let $\HH$ be a finite-dimensional Hilbert space, and $\AA_1,\BB_1\subseteq \D(\HH)$ two non-empty closed sets of quantum states on $\HH$. Using the notation in~\eqref{q_F_n_iid}--\eqref{q_F_n_av} and the definition~\eqref{q_Stein}, the Stein exponents of the two tasks where the alternative hypothesis is composite i.i.d.\ with base set $\BB_1$ and the null hypothesis is either composite i.i.d.\ or arbitrarily varying with base set $\AA_1$ are given by
\begin{align}
\rel{\stein}{\AA_1^\mathrm{iid}}{\BB_1^\mathrm{iid}} &= \minst_{\mu\in \PP(\BB_1)}\, \lim_{n\to\infty} \frac1n \inf_{\rho\in \AA_1} D\Big(\rho^{\otimes n}\, \Big\|\, \scaleobj{.8}{\int_{\BB_1}} \dd\mu(\sigma)\ \sigma^{\otimes n}\Big)\, , \label{q_iid_vs_iid} \\
\rel{\stein}{\AA_1^\mathrm{av}}{\BB_1^\mathrm{iid}} &= \minst_{\mu\in \PP(\BB_1)}\, \lim_{n\to\infty} \frac1n \inf_{\rho_n \in \co\scaleobj{1.3}{(}\AA_1^{\otimes n,\, \mathrm{av}}\scaleobj{1.3}{)}} D\Big(\rho_n\, \Big\|\, \scaleobj{.8}{\int_{\BB_1}} \dd\mu(\sigma)\ \sigma^{\otimes n}\Big)\, ,
\label{q_av_vs_iid}
\end{align}
where $\PP(\BB_1)$ is the set of all probability measures on the compact set $\BB_1$, and $\minst_\mu$ indicates that we restrict the minimisation to those $\mu$ such that the inner limit in $n$ exists (such a set is non-empty). If $\BB_1$ is also convex, we can rewrite~\eqref{q_iid_vs_iid} more simply by pulling the minimisation over $\rho\in \AA_1$ out of the limit:
\bb
\rel{\stein}{\AA_1^\mathrm{iid}}{\BB_1^\mathrm{iid}} &= \minst_{\substack{\rho\in \AA_1, \\[0pt] \mu\in \PP(\BB_1)}} \lim_{n\to\infty} \frac1n\, D\Big(\rho^{\otimes n}\, \Big\|\, \scaleobj{.8}{\int_{\BB_1}} \dd\mu(\sigma)\ \sigma^{\otimes n}\Big)\, . \label{q_iid_vs_iid_convexity}
\ee

In the case where the alternative hypothesis is arbitrarily varying, the Stein exponents are given by the exact same expressions, but with $\BB_1$ replaced by its convex hull $\co(\BB_1)$. Formally,
\bb
\rel{\stein}{\AA_1^\mathrm{iid}}{\BB_1^\mathrm{av}} &= \rel{\stein}{\AA_1^\mathrm{iid}}{\co(\BB_1)^\mathrm{av}} =  \rel{\stein}{\AA_1^\mathrm{iid}}{\co(\BB_1)^\mathrm{iid}} \\
&= \minst_{\substack{\rho\in \AA_1, \\[0pt] \mu\in \PP(\co(\BB_1))}} \lim_{n\to\infty} \frac1n\, D\Big(\rho^{\otimes n}\, \Big\|\, \scaleobj{.8}{\int_{\co(\BB_1)}} \dd\mu(\sigma)\ \sigma^{\otimes n}\Big)\, , \label{q_iid_vs_av} 
\ee
\vspace{-1.5ex}
\bb
\rel{\stein}{\AA_1^\mathrm{av}}{\BB_1^\mathrm{av}} &= \rel{\stein}{\AA_1^\mathrm{av}}{\co(\BB_1)^\mathrm{av}} = \rel{\stein}{\AA_1^\mathrm{av}}{\co(\BB_1)^\mathrm{iid}} \\
&= \minst_{\mu\in \PP(\co(\BB_1))}\, \lim_{n\to\infty} \frac1n \inf_{\rho_n \in \co\scaleobj{1.3}{(}\AA_1^{\otimes n,\, \mathrm{av}}\scaleobj{1.3}{)}} D\Big(\rho_n\, \Big\|\, \scaleobj{.8}{\int_{\co(\BB_1)}} \dd\mu(\sigma)\ \sigma^{\otimes n}\Big)\, .
%\lim_{n\to\infty} \inf_{\substack{\rho_n \in \co\scaleobj{1.3}{(}\AA_1^{\otimes n,\, \mathrm{av}}\scaleobj{1.3}{)}, \\ \mu_n \in \PP(\co(\BB_1))}} D\Big(\rho_n\, \Big\|\, \scaleobj{.8}{\int_{\co(\BB_1)}} \dd\mu_n(\sigma_1)\ \sigma_1^{\otimes n}\Big)\, .
\label{q_av_vs_av}
\ee
\end{thm}

The following is an instructive immediate consequence of the above result. The proofs of both Theorem~\ref{q_both_composite_iid_or_av_thm} and Corollary~\ref{q_iid_vs_both_cor} are presented in Section~\ref{subsec_proofs_both_composite}.

\begin{cor} \label{q_iid_vs_both_cor}
Let $\HH$ be a finite-dimensional Hilbert space, and $\AA_1,\BB_1\subseteq \D(\HH)$ two non-empty closed sets of quantum states on $\HH$, with $\BB_1$ convex. The Stein exponents of the two tasks where the null hypothesis is composite i.i.d.\ with base set $\AA_1$ and the alternative hypothesis is either composite i.i.d.\ or arbitrarily varying with base set $\BB_1$ are equal, and given by
\bb
\rel{\stein}{\AA_1^\mathrm{iid}}{\BB_1^\mathrm{iid}} = \rel{\stein}{\AA_1^\mathrm{iid}}{\BB_1^\mathrm{av}} = \minst_{\substack{\rho\in \AA_1, \\[0pt] \mu\in \PP(\BB_1)}} \lim_{n\to\infty}\, \frac1n\, D\Big(\rho^{\otimes n}\, \Big\|\, \scaleobj{.8}{\int_{\BB_1}} \dd\mu(\sigma)\ \sigma^{\otimes n}\Big)\, ,
\label{q_iid_vs_both}
\ee
where $\PP(\BB_1)$ is the set of probability measures on $\BB_1$, and, as before, the minimisation is restricted to the  non-empty set of pairs $(\rho,\mu)$ such that the inner limit in $n$ exists.
\end{cor}

It is instructive to compare Corollary~\ref{q_iid_vs_both_cor} to~\cite[Theorem~1.1]{berta_composite}. The setting of interest there is that of two closed and convex sets $\AA_1,\BB_1\subseteq \D(\HH)$. With some further (minor) assumptions on the supports, in~\cite[Theorem~1.1]{berta_composite} it is shown that
\bb
\rel{\stein}{\AA_1^\mathrm{iid}}{\BB_1^\mathrm{iid}} = \lim_{n\to\infty} \frac1n \inf_{\substack{\rho\in \AA_1, \\[0pt] \mu_n\in \PP(\BB_1)}} D\Big(\rho^{\otimes n}\, \Big\|\, \scaleobj{.8}{\int_{\BB_1}} \dd\mu_n(\sigma)\ \sigma^{\otimes n}\Big)\, .
\label{BBH_result}
\ee
Our~\eqref{q_iid_vs_both} constitutes an improvement over~\eqref{BBH_result} in three different ways:
\begin{itemize}[leftmargin=20pt]
\item First, because there is no convexity assumption on $\AA_1$, nor is the support condition in~\cite[Eq.~(13)]{berta_composite} needed.
\item Secondly, because~\eqref{q_iid_vs_both} gives an expression for both $\rel{\stein}{\AA_1^\mathrm{iid}}{\BB_1^\mathrm{iid}}$ and $\rel{\stein}{\AA_1^\mathrm{iid}}{\BB_1^\mathrm{av}}$. The results of~\cite{berta_composite}, on the contrary, cover only the former case.
\item Thirdly, and more importantly, because in~\eqref{q_iid_vs_both} the minimisations over states $\rho\in \AA_1$ and measures $\mu\in \PP(\BB_1)$ are outside the limit, instead of inside like in~\eqref{BBH_result}. Besides being computationally more convenient, as we eliminated the need to optimise over $\rho$ and $\mu$ separately for each $n$, Eq.~\eqref{q_iid_vs_both} is also an improvement from the information theoretic perspective. Indeed, the rightmost side of~\eqref{q_iid_vs_both} is manifestly at least as large as that of~\eqref{BBH_result}, i.e.\ it corresponds to a quantum hypothesis testing procedure that is at least as good as the one that is implicit in~\eqref{BBH_result}. Since in all of these problems the non-trivial statement is always achievability, Eq.~\eqref{q_iid_vs_both} can be truly considered as an improvement over~\eqref{BBH_result}. (As we will see, we will also recover~\eqref{BBH_result} in the course of our proof; cf.\ the second line of~\eqref{q_both_composite_iid_or_av_proof_eq7}).
\end{itemize}

It is also possible to compare~\eqref{q_av_vs_av} with~\cite[Theorem~25]{Fang2025}. Looking at the list of axioms used there, reported in~\cite[Assumption~24]{Fang2025}, we see that the two sequences of hypotheses $\co\big(\AA_1^\mathrm{av}\big)$ and $\co\big(\BB_1^\mathrm{av}\big)$ satisfy them. (This is not the case, instead, if in either hypothesis we replace $\mathrm{av}$ with $\mathrm{iid}$.) Hence, using~\cite[Theorem~25]{Fang2025} we get that
\bb
\rel{\stein}{\AA_1^\mathrm{av}}{\BB_1^\mathrm{av}} = \rel{D^\infty}{\co\big(\AA_1^\mathrm{av}\big)}{\co\big(\BB_1^\mathrm{av}\big)} = \lim_{n\to\infty} \frac1n\, \inf_{\substack{\rho_n \in \co\scaleobj{1.1}{(}\AA_1^{\otimes n,\, \mathrm{av}}\scaleobj{1.1}{)} \\ \sigma_n \in \co\scaleobj{1.1}{(}\BB_1^{\otimes n,\, \mathrm{av}}\scaleobj{1.1}{)}}} D(\rho_n\|\sigma_n)\, . \\[-1ex]
\ee
This is comparable to our expression on the second line of~\eqref{q_av_vs_av}, which is, however, markedly simpler.

\section{Notation} \label{sec_notation}

Before presenting the proofs of the above results, we need to fix some basic notation. A state on a quantum system is represented by a \deff{density operator}, i.e.\ a positive semi-definite trace class operator with trace one, on a Hilbert space $\HH$. Here we will only consider finite-dimensional Hilbert spaces; the set of density operators on $\HH$ is denoted as $\D(\HH)$. A \deff{quantum measurement} on the system represented by $\HH$, also called a \deff{POVM}, is a finite collection $(E_x)_{x\in \XX}$ of positive semi-definite operators $E_x\geq 0$ such that $\sum_{x\in \XX} E_x = \id$. We can think of measurements as maps $\MM:\D(\HH)\to\PP(\XX)$ that take as input a quantum state and output a classical probability distribution, defined by $(\MM(\rho))(x) \coloneqq \Tr[\rho E_x]$.

The task of \deff{quantum hypothesis testing} can be defined by following the discussion in Section~\ref{subsec_background}: given some Hilbert space $\HH$ and two sequences $\AA = (\AA_n)_n$ and $\BB = (\BB_n)_n$ of sets $\AA_n,\BB_n\subseteq \D\big(\HH^{\otimes n}\big)$, at step $n$ we are handed over a density operator $\rho_n \in \D\big(\HH^{\otimes n}\big)$. Our goal is to guess, by means of a suitable binary quantum measurement $(E_n,\id-E_n)$, whether $\rho_n \in \AA_n$ (null hypothesis $\mathrm{H}_0$) or $\rho_n \in \BB_n$ (alternative hypothesis $\mathrm{H}_1$), given the promise that one of these two options is correct. To investigate the ultimate physical limits of quantum hypothesis testing, we will assume that any quantum measurement on $\HH^{\otimes n}$ is achievable. The case where only a restricted set of measurements is available has also been studied~\cite{brandao_adversarial}.

Mistaking $\mathrm{H}_0$ for $\mathrm{H}_1$ is a type I error, while mistaking $\mathrm{H}_1$ for $\mathrm{H}_0$ is a type II error. Therefore, the minimal type II error probability for a given threshold $\e\in (0,1)$ on the type I error probability can then be written as (cf.~\eqref{beta_e_level_n})
\bb
\beta_\e(\AA_n\|\BB_n) = \inf\left\{ \sup_{\sigma_n\in \BB_n}\! \Tr \sigma_n E_n:\ \ 0\leq E_n\leq \id,\ \, \sup_{\rho_n\in \AA_n}\! \Tr \rho_n (\id-E_n) \leq \e \right\} ,
\label{q_beta_e_level_n}
\ee
where the operators $E_n$ act on $\HH^{\otimes n}$, and we recalled that, for two operators $X,Y$, the inequality $X\leq Y$ means that $Y-X$ is positive semi-definite. We can now introduce the \deff{hypothesis testing relative entropy}~\cite{Buscemi2010}
\bb
D_H^\e(\rho\|\sigma) \coloneqq -\log \inf\left\{ \Tr \sigma E:\ \ 0\leq E\leq \id,\ \, \Tr \rho (\id-E) \leq \e \right\} ,
\label{hypothesis_testing_relent}
\ee
in terms of which the negative logarithm of~\eqref{q_beta_e_level_n} can be re-written as~\cite[Lemma~31]{Fang2025}
\bb
-\log \beta_\e(\AA_n\|\BB_n) = \rel{D_H^\e}{\co(\AA_n)}{\co(\BB_n)}\, ,
\label{D_H_and_beta_e}
\ee
where, to define the right-hand side, we adopted the following convention: given a function $\mathds{D}(\cdot\|\cdot): \D(\HH) \times \D(\HH) \to \R$ on pairs of states and two sets $\AA_1,\BB_1\subseteq \D(\HH)$, we set
\bb
\mathds{D}(\AA_1\|\BB_1) \coloneqq \inf_{\rho\in \AA_1,\, \sigma\in \BB_1} \mathds{D}(\rho\|\sigma)\, .
\label{q_relent_sets}
\ee

Now, the \deff{Stein exponent} corresponding to the above quantum hypothesis testing task can then be constructed as (cf.~\eqref{Stein})
\bb
\stein(\AA\|\BB) &= \lim_{\e\to 0^+} \liminf_{n\to\infty} \left\{ -\frac1n\log \beta_\e(\AA_n\|\BB_n)\right\} \\
&= \lim_{\e\to 0^+} \liminf_{n\to\infty} \frac1n\, \rel{D_H^\e}{\co(\AA_n)}{\co(\BB_n)}\, .
\label{q_Stein}
\ee
We observe immediately that the Stein exponent is unchanged if we convexify every element of either of the two sequences of hypotheses. Adopting the intuitive convention of defining
\bb
\co(\FF) \coloneqq \big(\co(\FF_n)\big)_n
\label{q_convexify_sequence_sets}
\ee
for a sequence of sets $\FF = (\FF_n)_n$, we can express this property as the series of identities
\bb
\stein(\AA\|\BB) &= \stein(\co(\AA)\|\BB) = \stein(\AA\|\co(\BB)) = \stein(\co(\AA)\|\co(\BB))\, .
\label{convexify_Stein}
\ee

The (Umegaki) \deff{relative entropy} between two quantum states $\rho,\sigma\in \D(\HH)$ is given by~\cite{Umegaki1962} (cf.~\eqref{Hiai_Petz})
\bb
D(\rho\|\sigma) = \Tr \left[\rho \left(\log \rho - \log \sigma\right) \right] ,
\label{Umegaki}
\ee
where we agree by convention that $D(\rho\|\sigma) = +\infty$ if $\supp(\rho)\not\subseteq \supp(\sigma)$, where the support $\supp(X)$ of a Hermitian operator $X$ is the span of the eigenvectors of $X$ corresponding to non-zero eigenvalues. The operational importance of~\eqref{Umegaki} rests on the fact that it captures precisely the Stein exponent between two simple i.i.d.\ hypotheses, in the sense that, for any $\rho,\sigma\in \D(\HH)$, it holds that~\cite{Hiai1991, Ogawa2000, Bjelakovic2003} (cf.~\eqref{Hiai_Petz})
\bb
\rel{\stein}{(\rho^{\otimes n})_n}{(\sigma^{\otimes n})_n} = D(\rho\|\sigma)\, .
\label{q_Hiai_Petz}
\ee

Given two sets $\AA_1,\BB_1\subseteq \D(\HH)$, their relative entropy $D(\AA_1\|\BB_1)$ is defined according to~\eqref{q_relent_sets}. Analogously, for two \emph{sequences} $\AA = (\AA_n)_n$ and $\BB = (\BB_n)_n$ of sets $\AA_n,\BB_n\subseteq \D\big(\HH^{\otimes n}\big)$, we stipulate that
\bb
D^\infty(\AA\|\BB) \coloneqq \liminf_{n\to\infty} \frac1n\, D(\AA_n\|\BB_n) = \liminf_{n\to\infty} \frac1n\, \inf_{\rho_n \in \AA_n,\ \sigma_n \in \BB_n} D(\rho_n\|\sigma_n)\, .
\label{q_regularised_relent_seq}
\ee

In analysing the right-hand side of~\eqref{q_regularised_relent_seq}, the following well-known result is often useful. It allows to simplify the optimisation to permutationally symmetric pairs of states $(\rho_n,\sigma_n)$ in the case where both $\AA$ and $\BB$ satisfy Axiom~\ref{q_ax_closed_permutations}.

\begin{lemma} \label{relent_perm_symm_lemma}
Let $\HH$ be a Hilbert space, and consider two sequences $\AA = (\AA_n)_n$ and $\BB = (\BB_n)_n$ of sets of states $\AA_n,\BB_n\subseteq \D\big(\HH^{\otimes n}\big)$ that obey Axiom~\ref{q_ax_closed_permutations}. For some $n\in\N^+$, let $\mathds{D}(\cdot\|\cdot):\D\big(\HH^{\otimes n}\big)\times \D\big(\HH^{\otimes n}\big) \to [0,\infty)$ be a function defined on pairs of states. Assume that $\mathds{D}(\cdot\|\cdot)$ is:
\begin{enumerate}[(a),itemsep=0.5ex]
\item jointly convex, i.e.\ such that $\rel{\mathds{D}}{\sumno_x p_x \rho_{n,x}}{\sumno_x p_x \sigma_{n,x}} \leq \sumno_x p_x \mathds{D}(\rho_{n,x}\|\sigma_{n,x})$ for all finite collections of states $(\rho_{n,x})_x$ and $(\sigma_{n,x})_x$ on $\HH^{\otimes n}$ and all probability distributions $p = (p_x)_x$;
\item unitarily invariant, in the sense that $\rel{\mathds{D}}{U_n^{\vphantom{\dag}} \rho_n U_n^\dag}{U_n^{\vphantom{\dag}} \rho_n U_n^\dag} = \mathds{D}(\rho_n\|\sigma_n)$ for all $\rho_n,\sigma_n\in \D\big(\HH^{\otimes n}\big)$ and all unitaries $U_n$ on $\HH^{\otimes n}$.
\end{enumerate}
Then, the infimum in
\bb
\mathds{D}(\AA_n\|\BB_n) = \inf_{\rho_n \in \AA_n,\ \sigma_n \in \BB_n} \mathds{D}(\rho_n\|\sigma_n)
\ee
can be restricted to permutationally invariant states. In other words, we can assume without loss of generality that $\rho_n$ and $\sigma_n$ satisfy that $U_\pi^{\vphantom{\dag}} \rho_n U_\pi^\dag = \rho_n$ and $U_\pi^{\vphantom{\dag}} \sigma_n U_\pi^\dag = \sigma_n$ for all $\pi\in S_n$, where $U_\pi$ is the unitary that acts by permuting the tensor factors of $\HH^{\otimes n}$ according to $\pi$.
\end{lemma}

\begin{proof}
In fact, for any pair of states $\rho'_n \in \AA_n$ and $\sigma'_n \in \BB_n$, we can set
\bb
\rho_n \coloneqq \E_\pi U_\pi^{\vphantom{\dag}} \rho'_n U_\pi^\dag\, ,\qquad \sigma_n \coloneqq \E_\pi U_\pi^{\vphantom{\dag}} \sigma'_n U_\pi^\dag\, ,
\ee
where $\pi\in S_n$ is drawn uniformly at random, and write
\bb
\mathds{D}(\rho_n\|\sigma_n) \leq \E_\pi\, \rel{\mathds{D}}{U_\pi^{\vphantom{\dag}} \rho'_n U_\pi^\dag}{U_\pi^{\vphantom{\dag}} \sigma'_n U_\pi^\dag} = \mathds{D}(\rho'_n\|\sigma'_n)\, ,
\ee
where the inequality is by joint convexity of the measured relative entropy, and the equality by unitary invariance. Since $\rho_n\in \AA_n$ and $\sigma_n \in \BB_n$, due to the convexity and closure under permutations of the respective sets, and both states are permutationally invariant by construction, we have shown that the pair $(\rho_n,\sigma_n)$ attains a value of the function $\mathds{D}(\cdot\|\cdot)$ that is at least as small as that attained by $(\rho'_n,\sigma'_n)$. This proves the claim.
\end{proof}

The regularised relative entropy between sequences of sets appears in the following well-known general converse result, which can be proved 
%quantum version of Lemma~\ref{converse_double_Stein_lemma} is proved in the exact same way, 
using the data processing inequality for the Umegaki relative entropy~\cite{lieb73a, lieb73b, lieb73c, Lindblad-monotonicity}. See~\cite[Proposition~2.1]{berta_composite} for details. 

\begin{lemma} \label{q_converse_double_Stein_lemma}
For a finite-dimensional Hilbert space $\HH$, let $\AA = (\AA_n)_n$ and $\BB = (\BB_n)_n$ be two sequences of sets of states $\AA_n,\BB_n\subseteq \D\big(\HH^{\otimes n}\big)$. Then we have
\bb
\stein(\AA\|\BB) \leq D^\infty(\co(\AA)\|\co(\BB))\, ,
\label{q_converse_double_Stein}
\ee
where we adopted the conventions in~\eqref{q_convexify_sequence_sets} and~\eqref{q_regularised_relent_seq}.
\end{lemma}

\section{Proofs} \label{sec_proofs}

This section is devoted to the presentation of the complete proofs of all of our results.

\subsection{Some properties of regularised relative entropies between sequences of sets}

We start by laying the groundwork for the proof of Theorem~\ref{stronger_genq_Sanov_thm}. An important result in this sense is the following proposition, which establishes a rather surprising connection between the regularised relative entropies corresponding to two seemingly different settings, where the alternative hypothesis is either composite i.i.d.\ or arbitrarily varying. 

\begin{prop} \label{av_to_iid_reduction_reg_relent_prop}
Let $\HH$ be a finite-dimensional Hilbert space, $\BB_1\subseteq \D(\HH)$ a closed and convex set of states, and $\AA = (\AA_n)_n$ any sequence of convex sets $\AA_n\subseteq \D\big(\HH^{\otimes n}\big)$ that obeys Axiom~\ref{q_ax_closed_permutations}. Set, as usual,
\bb
\rel{D^\infty}{\AA}{\co\big(\BB_1^\mathrm{b}\big)} = \liminf_{n\to\infty} \frac1n\, \rel{D}{\AA_n}{\co\big(\BB_1^{\otimes n,\,\mathrm{b}}\big)}\, ,\qquad \mathrm{b}\in \{\mathrm{iid},\mathrm{av}\}\, .
\label{av_to_iid_reduction_reg_relent_definitions}
\ee
Then:
\begin{enumerate}[1.]
\item The value of~\eqref{av_to_iid_reduction_reg_relent_definitions} is the same for $\mathrm{b} = \mathrm{iid}$ and $\mathrm{b} = \mathrm{av}$, meaning that
\bb
\rel{D^\infty}{\AA}{\co\big(\BB_1^\mathrm{iid}\big)} = \rel{D^\infty}{\AA}{\co\big(\BB_1^\mathrm{av}\big)}\, .
\label{av_to_iid_reduction_reg_relent}
\ee
\item If any of the two limit infima in~\eqref{av_to_iid_reduction_reg_relent_definitions}  can be replaced with an ordinary limit, then so is true of the other (and, of course, the two limits are equal).
\item This happens, for example, if $\AA$ is closed under tensor products, i.e.\ if $\rho_n\otimes \rho_m\in \AA_{n+m}$ for all $\rho_n\in \AA_n$ and $\rho_m\in \AA_m$.
\end{enumerate}
\end{prop}

\begin{proof}
We start with the first claim. In the early part of the proof, we apply a discretisation procedure to $\BB_1$, so as to effectively reduce ourselves to the case of a finite $\BB_1$. We will then solve the latter by employing types. We remind the reader that, given a finite alphabet $\XX$ and some $n\in \N^+$, an \deff{$n$-type} on $\XX$ (or, simply, a \deff{type}) is a probability distribution $V:\XX\to [0,1]$ on $\XX$ with the property that $nV(x)\in \N$ for all $x\in \XX$. We denote the set of $n$-types on $\XX$ as $\TT_n$.

We begin with some preliminary considerations. Without loss of generality, we can assume that $\BB_1$ is non-empty. Since it is a convex set, elementary considerations show that all states in its relative interior, which is also non-empty and will be denoted by $\relint(\BB_1)$, must have the same support, call it $\KK \subseteq \HH$, and that $\KK$ must contain the support of any other state in $\BB_1$. Therefore, any state $\sigma_n \in \co\big(\BB_1^{\otimes n,\,\mathrm{b}}\big)$ will also satisfy $\supp(\sigma_n) \subseteq \KK^{\otimes n}$, entailing that the infimum over $\rho_n\in \AA_n$ that is implicit in the right-hand side of~\eqref{av_to_iid_reduction_reg_relent_definitions} can be restricted to states $\rho_n$ whose support lies also in $\KK$. Hence, up to considering a smaller $\HH$, we can assume without loss of generality that $\KK=\HH$. If we do that, then we will automatically have that $\sigma>0$ for all $\sigma \in \relint(\BB_1)$.

Now, let us fix some $\delta>0$. For some $\sigma\in \relint(\BB_1)$, consider the set of operators
\bb
\widetilde{\pazocal{U}}_\sigma^\delta \coloneqq \left\{ Z=
Z^\dag:\ Z < \exp[\delta] \,\sigma\right\} = \left\{ Z=Z^\dag:\ \big\|\sigma^{-1/2} Z \sigma^{-1/2}\big\|_\infty < \exp[\delta] \right\} ,
\ee
where $\|\cdot\|_\infty$ is the operator norm. Since $\|\cdot\|_\infty$ is a continuous function, the set $\widetilde{\pazocal{U}}^\delta_\sigma$ is open (as a subset of the real Euclidean space of all Hermitian operators). We now claim that
\bb
\BB_1 \subseteq \bigcup_{\sigma\in \relint(\BB_1)} \widetilde{\pazocal{U}}^\delta_\sigma\, .
\label{av_to_iid_reduction_reg_relent_proof_eq2}
\ee
In fact, consider an arbitrary $\sigma\in \BB_1$, and pick any $\tau\in \relint(\BB_1)$ (as said, this latter set is non-empty because $\BB_1$ is convex and non-empty). It is immediate to verify that $\tau' \coloneqq (1-p) \tau + p\sigma \in \relint(\BB_1)$ for all $p\in [0,1)$. Taking $p = \exp[-\delta]$ and using the fact that $\tau>0$, we have
\bb
\sigma < \frac1p \left( (1-p) \tau + p\sigma \right) = \exp[\delta]\, \tau'\, ,
\ee
entailing that $\sigma \in \widetilde{\pazocal{U}}^\delta_{\tau'} \subseteq \bigcup_{\sigma'\in \relint(\BB_1)} \widetilde{\pazocal{U}}^\delta_{\sigma'}$; since $\sigma\in \BB_1$ was arbitrary, this proves~\eqref{av_to_iid_reduction_reg_relent_proof_eq2}.

To proceed further, note that $\BB_1$, which is a closed set of states and hence also bounded, is compact. Therefore, from the open cover in~\eqref{av_to_iid_reduction_reg_relent_proof_eq2}  we can extract a finite sub-cover
\bb
\BB_1 \subseteq \bigcup_{x\in \XX_\delta} \,\widetilde{\pazocal{U}}^\delta_{\sigma_x}\, ,
\label{av_to_iid_reduction_reg_relent_proof_eq4}
\ee
where $\sigma_x \in \relint(\BB_1)$ for all $x\in \XX_\delta$, and $|\XX_\delta|<\infty$. Enumerating the elements of $\XX_\delta$ as $x_1,\ldots, x_N$, we set
\bb
\pazocal{U}^\delta_{\sigma_{x_1}} \coloneqq \widetilde{\pazocal{U}}_{\sigma_{x_1}}^\delta \cap \BB_1\, ,\qquad  \pazocal{U}^\delta_{\sigma_{x_i}} \coloneqq \widetilde{\pazocal{U}}^\delta_{\sigma_{x_i}}\cap \BB_1\cap \left( \bigcup\nolimits_{j=1}^{i-1} \widetilde{\pazocal{U}}_{\sigma_{x_j}} \right)^c\qquad (i=2,\ldots,N)\, .
\ee
Note that, for all $x\in \XX_\delta$, we have $\pazocal{U}^\delta_{\sigma_{x}}\subseteq \widetilde{\pazocal{U}}^\delta_{\sigma_{x}}$, and hence also $\sigma < \exp[\delta]\,\sigma_x$ for all $\sigma\in \pazocal{U}^\delta_{\sigma_{x}}$. These new sets $\pazocal{U}^\delta_{\sigma_{x}}$ contain only states and are still Borel (although, in general, not open any more). They constitute a partition of $\BB_1$, because they are disjoint by construction: formally, 
\bb
\BB_1 = \bigcup_{x\in \XX_\delta} \,\pazocal{U}^\delta_{\sigma_x}\, ,\qquad \sigma_x\in \BB_1\quad \forall\ x\in \XX_\delta\, ,\quad \pazocal{U}^\delta_{\sigma_x} \cap \pazocal{U}^\delta_{\sigma_y} = \emptyset\quad \forall\ x,y\in \XX_\delta :\ \ x\neq y\, .
\label{discretisation_partition_B_1}
\ee

We now move on to the proof of~\eqref{av_to_iid_reduction_reg_relent}. Due to the convexity and permutation invariance of both $\AA_n$ and $\co\big( \BB_1^{\otimes n,\, \mathrm{av}}\big)$, Lemma~\ref{relent_perm_symm_lemma} implies that the optimisation over pairs of states that is implicit in $\rel{D}{\AA_n}{\co\big( \BB_1^{\otimes n,\, \mathrm{av}}\big)}$ can be restricted to permutationally symmetric states. As one readily verifies, a permutationally symmetric state in $\co\big( \BB_1^{\otimes n,\,\mathrm{av}}\big)$ takes the form
\bb
\Omega_n = \sum_{i=1}^M p_i\, \E_\pi \left[ \omega_{i,\pi(1)} \otimes \ldots \otimes \omega_{i,\pi(n)} \right] ,
\label{av_to_iid_reduction_reg_relent_proof_eq6}
\ee
where $M$ is finite (and can be bounded by Carath\'eodory's theorem), $\omega_{i,j}\in \BB_1$ for all $i=1,\ldots, M$ and $j=1,\ldots,n$, and $\pi\in S_n$ is a uniformly random permutation.  Therefore, we can write
\bb
D(\AA_n\|\,\Omega_n) \leq \rel{D}{\AA_n}{\co\big(\BB_1^{\otimes n,\,\mathrm{av}}\big)} + 1\, ,
\label{av_to_iid_reduction_reg_relent_proof_eq7}
\ee
for some $\Omega_n$ of the form~\eqref{av_to_iid_reduction_reg_relent_proof_eq6}. 

Due to~\eqref{discretisation_partition_B_1}, each of the states $\omega_{i,j}$ appearing in~\eqref{av_to_iid_reduction_reg_relent_proof_eq6} belongs to exactly one set $\pazocal{U}^\delta_{\sigma_x}$. Let us call $\widebar{x} :\{1\ldots, M\}\times \{1,\ldots,n\} \to \XX_\delta$ the function such that $\omega_{i,j} \in \pazocal{U}^\delta_{\sigma_{\widebar{x}(i,j)}}$ for all $i,j$. In particular,
\bb
\omega_{i,j} \leq \exp[\delta]\, \sigma_{\widebar{x}(i,j)}
\ee
(even with a strict inequality); plugging this into~\eqref{av_to_iid_reduction_reg_relent_proof_eq6}, we obtain that
\bb
\Omega_n \leq \exp[n\delta] \sum_{i=1}^M p_i\, \E_\pi \left[ \sigma_{\widebar{x}(i, \pi(1))} \otimes \ldots \otimes \sigma_{\widebar{x}(i, \pi(n))} \right] .
\label{av_to_iid_reduction_reg_relent_proof_eq9}
\ee
The state defined by the sum on the right-hand side belongs to $\co\big(\{\sigma_x\!:\, x\in \XX_\delta\}^{\otimes n,\,\mathrm{av}}\big)$ and is permutationally symmetric; therefore, it can be written as
\bb
\sum_{i=1}^M p_i\, \E_\pi \left[ \sigma_{\widebar{x}(i, \pi(1))} \otimes \ldots \otimes \sigma_{\widebar{x}(i, \pi(n))} \right] = \sum_{V\in \TT_n} p(V)\, \gamma_{n,V}\, ,
\ee
where the sum is over the set $\TT_n$ of $n$-types over $\XX$ (defined at the beginning of the proof), $p$ is a probability distribution on $\TT_n$, and we introduced the notation
\bb
\gamma_{n,V} \coloneqq \frac{1}{|T_{n,V}|} \sum_{x^n\in T_{n,V}} \sigma_{x_1}\otimes \ldots \otimes \sigma_{x_n}\, .
\ee
From~\eqref{av_to_iid_reduction_reg_relent_proof_eq9} we have
\bb
\Omega_n \leq \exp[n\delta] \sum_{V\in \TT_n} p(V)\, \gamma_{n,V}\, .
\label{av_to_iid_reduction_reg_relent_proof_eq12}
\ee

The crucial insight of the proof comes now. We observe that $\gamma_{n,V}$ can be upper bounded with a polynomial multiple of the i.i.d.\ state $\left(\sum_x V(x)\, \sigma_x\right)^{\otimes n}$, where $\sum_x V(x)\, \sigma_x\in \BB_1$. To see this, it suffices to expand the tensor product, retain only the sequences with type $V$, and apply Sanov's theorem~\cite[Exercise~2.12(a), p.~29]{CSISZAR-KOERNER}. More in detail,
\bb
\left(\sumno_{x\in \XX_\delta} V(x)\, \sigma_x\right)^{\otimes n} &= \sum_{x^n\in \XX_\delta^n} V^{\otimes n}(x^n)\, \sigma_{x_1}\otimes \ldots \otimes \sigma_{x_n} \\
&= \sum_{W\in \TT_n} V^{\otimes n}(T_{n,W})\,\gamma_{n,W} \\
&\geq V^{\otimes n}(T_{n,V})\,\gamma_{n,V} \\
&\geq \frac{\gamma_{n,V}}{(n+1)^{|\XX_\delta|}}\, .
\label{av_to_iid_reduction_reg_relent_proof_eq13}
\ee
Setting
\bb
\Omega'_n \coloneqq \sum_{V\in \TT_n} p(V) \left(\sumno_{x\in \XX_\delta} V(x)\, \sigma_x\right)^{\otimes n} \in \co\big(\BB_1^{\otimes n,\,\mathrm{iid}}\big)\, ,
\ee
where we remembered that $\BB_1$ is convex, and combining~\eqref{av_to_iid_reduction_reg_relent_proof_eq12} and~\eqref{av_to_iid_reduction_reg_relent_proof_eq13}, we see that
\bb
\Omega_n \leq (n+1)^{|\XX_\delta|}\, \exp[n\delta]\, \Omega'_n\, .
\ee
Plugging this inequality into~\eqref{av_to_iid_reduction_reg_relent_proof_eq7} and exploiting the operator monotonicity of the logarithm yields
\bb
\rel{D}{\AA_n}{\co\big(\BB_1^{\otimes n,\,\mathrm{iid}}\big)} &\leq D(\AA_n\|\,\Omega'_n) \\
&\leq D(\AA_n\|\,\Omega_n) + n\delta + |\XX_\delta| \log(n+1) \\
&\leq \rel{D}{\AA_n}{\co\big(\BB_1^{\otimes n,\,\mathrm{av}}\big)} + 1 + n\delta + |\XX_\delta| \log(n+1)\, .
\label{av_to_iid_reduction_reg_relent_proof_eq16}
\ee
We can now divide by $n$ and append also the trivial inequality 
%$\rel{D}{\AA_n}{\co\big(\BB_1^{\otimes n,\,\mathrm{av}}\big)} \leq \rel{D}{\AA_n}{\co\big(\BB_1^{\otimes n,\,\mathrm{iid}}\big)}$, which 
that follows from the inclusion relation $\BB_1^\mathrm{iid} \subseteq \BB_1^\mathrm{av}$, thus obtaining
\bb
\frac1n\,\rel{D}{\AA_n}{\co\big(\BB_1^{\otimes n,\,\mathrm{av}}\big)} &\leq \frac1n\,\rel{D}{\AA_n}{\co\big(\BB_1^{\otimes n,\,\mathrm{iid}}\big)} \\ &\leq \frac1n\,\rel{D}{\AA_n}{\co\big(\BB_1^{\otimes n,\,\mathrm{av}}\big)} + \delta + \frac{1+|\XX_\delta| \log(n+1)}{n}\, .
\label{av_to_iid_reduction_reg_relent_proof_eq17}
\ee
Taking the limit infimum as $n\to\infty$, one obtains the inequality
\bb
\rel{D^\infty}{\AA}{\co\big(\BB_1^{\mathrm{av}}\big)} \leq \rel{D^\infty}{\AA}{\co\big(\BB_1^\mathrm{iid}\big)} \leq \rel{D^\infty}{\AA}{\co\big(\BB_1^{\mathrm{av}}\big)} + \delta\, .
\ee
Since $\delta>0$ was arbitrary, we can now send $\delta\to 0^+$ and prove~\eqref{av_to_iid_reduction_reg_relent}.

The second claim follows once again from~\eqref{av_to_iid_reduction_reg_relent_proof_eq17}, which also implies that
\bb
\limsup_{n\to\infty} \frac1n\,\rel{D}{\AA_n}{\co\big(\BB_1^{\otimes n,\,\mathrm{av}}\big)} &= \limsup_{n\to\infty} \frac1n\,\rel{D}{\AA_n}{\co\big(\BB_1^{\otimes n,\,\mathrm{iid}}\big)}\, .
\ee

We now move on to the third claim. If $\AA$ is closed under tensor products, then the sequence $n\mapsto \rel{D}{\AA_n}{\co\big(\BB_1^{\otimes n,\,\mathrm{av}}\big)}$ turns out to be sub-additive.\footnote{A sequence $\N^+\ni n\mapsto a_n$ is called sub-additive if $a_{n+m}\leq a_n + a_m$ for all $n,m\in \N^+$.} Indeed, for all $n,m\in \N^+$ and all quadruples of states $\rho_n\in \AA_n$, $\rho_m\in \AA_m$, $\sigma_n\in \co\big(\BB_1^{\otimes n,\,\mathrm{av}}\big)$, and $\sigma_m\in \co\big(\BB_1^{\otimes m,\, \mathrm{av}}\big)$, since
\bb
\sigma_n\otimes \sigma_m \in \co\Big(\BB_1^{\otimes (n+m),\, \mathrm{av}}\Big)\, ,
\label{av_to_iid_reduction_reg_relent_proof_eq20}
\ee
as a little thought shows, we have
\bb
D\Big(\AA_{n+m} \,\Big\|\, \co\Big(\BB_1^{\otimes (n+m),\,\mathrm{av}}\Big)\Big) &\leq D(\rho_n\otimes \rho_m\,\|\, \sigma_n\otimes \sigma_m) = D(\rho_n\|\sigma_n) + D(\rho_m \| \sigma_m)\, ,
\label{av_to_iid_reduction_reg_relent_proof_eq21}
\ee
which proves sub-additivity once we take the infimum over $\rho_n$, $\rho_m$, $\sigma_n$, and $\sigma_m$. Due to Fekete's lemma~\cite{Fekete1923}, the limit $\lim_{n\to\infty} \frac1n\, \rel{D}{\AA_n}{\co\big(\BB_1^{\otimes n,\,\mathrm{av}}\big)}$ exists. Then, by the second claim, also $\lim_{n\to\infty} \frac1n\, \rel{D}{\AA_n}{\co\big(\BB_1^{\otimes n,\,\mathrm{iid}}\big)}$ must exist, and the two need to be equal.
\end{proof}

The following result allows us to pull the optimisation over measures out of the limit infimum in the definition of regularised relative entropy.

\begin{lemma} \label{taking_measure_out_lemma}
Let $\HH$ be a finite-dimensional Hilbert space, and $\BB_1\subseteq \D(\HH)$ a Borel subset of states on $\HH$. For any sequence $\AA = (\AA_n)_n$ of sets $\AA_n\subseteq \D\big(\HH^{\otimes n}\big)$, the regularised relative entropy
\bb
\rel{D^\infty}{\AA}{\co\big(\BB_1^\mathrm{iid}\big)} &= \liminf_{n\to\infty} \frac1n\, \rel{D}{\AA_n}{\co\big(\BB_1^{\otimes n,\, \mathrm{iid}}\big)} \\
&= \liminf_{n\to\infty} \frac1n \inf_{\substack{\\[-1pt] \rho_n\in \AA_n,\\[0pt] \mu_n\in \PP(\BB_1)}} D\Big(\rho_n\, \Big\|\, \scaleobj{.8}{\int_{\BB_1}} \dd\mu_n(\sigma_1)\ \sigma_1^{\otimes n}\Big)
\label{taking_measure_out_prelim}
\ee
can be alternatively written by taking the infimum over all probability measures on $\BB_1$ outside the limit infimum, which has the added advantage of turning it into a minimum:
\bb
\rel{D^\infty}{\AA}{\co\big(\BB_1^\mathrm{iid}\big)} = \min_{\mu\in \PP(\BB_1)} \liminf_{n\to\infty} \frac1n \inf_{\rho_n\in \AA_n} D\Big(\rho_n\, \Big\|\, \scaleobj{.8}{\int_{\BB_1}} \dd\mu(\sigma_1)\ \sigma_1^{\otimes n}\Big)\, .
\label{taking_measure_out}
\ee
Moreover, if the limit infima in~\eqref{taking_measure_out_prelim} can be replaced by ordinary limits, then we have also
\bb
\rel{D^\infty}{\AA}{\co\big(\BB_1^\mathrm{iid}\big)} = \minst_{\mu\in \PP(\BB_1)}\, \lim_{n\to\infty} \frac1n \inf_{\rho_n\in \AA_n}\, D\Big(\rho_n\, \Big\|\, \scaleobj{.8}{\int_{\BB_1}} \dd\mu(\sigma_1)\ \sigma_1^{\otimes n}\Big)\, ,
\label{taking_measure_out_ordinary_limit}
\ee
where $\minst_\mu$ indicates that the minimum is restricted to those $\mu$ such that the inner limit in $n$ exists (such a set is non-empty).
\end{lemma}

\begin{proof}
By taking as ansatz for $\mu_n$ a fixed probability measure $\mu\in \PP(\BB_1)$, we see immediately that
\bb
\rel{D^\infty}{\AA}{\co\big(\BB_1^\mathrm{iid}\big)} \leq \inf_{\mu\in \PP(\BB_1)} \liminf_{n\to\infty} \frac1n \inf_{\rho_n\in \AA_n} D\Big(\rho_n\, \Big\|\, \scaleobj{.8}{\int_{\BB_1}} \dd\mu(\sigma_1)\ \sigma_1^{\otimes n}\Big)\, .
\label{taking_measure_out_proof_eq0}
\ee
The non-trivial inequality, therefore, is the opposite one. For each $n\in \N^+$, consider some $\mu_n\in \PP(\BB_1)$ such that
\bb
D\Big(\AA_n\, \Big\|\, \scaleobj{.8}{\int_{\BB_1}} \dd\mu_n(\sigma_1)\ \sigma_1^{\otimes n}\Big) \leq \inf_{\nu_n\in \PP(\BB_1)} D\Big(\AA_n\, \Big\|\, \scaleobj{.8}{\int_{\BB_1}} \dd\nu_n(\sigma_1)\ \sigma_1^{\otimes n}\Big) + 1 = \rel{D}{\AA_n}{\co\big(\BB_1^{\otimes n,\, \mathrm{iid}}\big)} + 1\, .
\label{taking_measure_out_proof_eq1}
\ee

Then, define a probability measure $\mu\in \PP(\BB_1)$ as
\bb
\mu \coloneqq \sum_{n=1}^\infty \frac{6}{\pi^2 n^2}\, \mu_n\, .
\ee
The above series is well defined because the space of regular signed measures on the $\sigma$-algebra of Borel subsets of $\BB_1$ is a Banach space, and hence complete, with respect to the total variation norm. In our case, the partial sums of the above series form a Cauchy sequence with respect to this norm, and hence converge to a limit measure that we call $\mu$. The non-negativity of $\mu$ is elementary, and, due to Euler's solution of the Basel problem~\cite{Euler-Basel}, $\mu$ is actually a probability measure. 

Naturally, for any $n\in \N^+$ it holds that
\bb
\int_{\BB_1} \dd\mu(\sigma_1)\ \sigma_1^{\otimes n} \geq \frac{6}{\pi^2 n^2} \int_{\BB_1} \dd\mu_n(\sigma_1)\ \sigma_1^{\otimes n} ,
\ee
entailing that
\bb
\frac1n\, D\Big(\AA_n\, \Big\|\, \scaleobj{.8}{\int_{\BB_1}} \dd\mu(\sigma_1)\ \sigma_1^{\otimes n}\Big) &\leqt{(i)} \frac1n\, D\Big(\AA_n\, \Big\|\, \frac{6}{\pi^2 n^2} \scaleobj{.8}{\int_{\BB_1}} \dd\mu_n(\sigma_1)\ \sigma_1^{\otimes n}\Big) \\
&= \frac1n\, D\Big(\AA_n\, \Big\|\, \scaleobj{.8}{\int_{\BB_1}} \dd\mu_n(\sigma_1)\ \sigma_1^{\otimes n}\Big) + \frac1n \log \frac{\pi^2 n^2}{6} \\
&\leqt{(ii)} \frac1n\, \rel{D}{\AA_n}{\co\big(\BB_1^{\otimes n,\, \mathrm{iid}}\big)} + \frac1n  + \frac1n \log \frac{\pi^2 n^2}{6}\, .
\label{taking_measure_out_proof_eq4}
\ee
Here, in~(i) we used the operator monotonicity of the logarithm, while~(ii) follows from~\eqref{taking_measure_out_proof_eq1}. 
Taking the limit infimum as $n\to\infty$, we obtain that
\bb
\inf_{\nu\in \PP(\BB_1)} \liminf_{n\to\infty} \frac1n\, D\Big(\AA_n\, \Big\|\, \scaleobj{.8}{\int_{\BB_1}} \dd\nu(\sigma_1)\ \sigma_1^{\otimes n}\Big) &\leq \liminf_{n\to\infty} \frac1n\, D\Big(\AA_n\, \Big\|\, \scaleobj{.8}{\int_{\BB_1}} \dd\mu(\sigma_1)\ \sigma_1^{\otimes n}\Big) \\
&\leq \liminf_{n\to\infty} \frac1n\left( \rel{D}{\AA_n}{\co\big(\BB_1^{\otimes n,\, \mathrm{iid}}\big)} + 1 + \log \frac{\pi^2 n^2}{6}\right) \\
&= \liminf_{n\to\infty} \frac1n\,\rel{D}{\AA_n}{\co\big(\BB_1^{\otimes n,\, \mathrm{iid}}\big)} \\
&= \rel{D^\infty}{\AA}{\BB_1^\mathrm{iid}}\, ,
\ee
which, together with~\eqref{taking_measure_out_proof_eq0}, shows that $\mu$ achieves the infimum on the leftmost side. This proves~\eqref{taking_measure_out}.

As for the last claim, we can reason as follows. If the limit infimum in the definition of $\rel{D^\infty}{\AA}{\co\big(\BB_1^\mathrm{iid}\big)}$ is actually a limit, then from~\eqref{taking_measure_out_proof_eq4} it also follows that
\bb
\limsup_{n\to\infty} \frac1n\, D\Big(\AA_n\, \Big\|\, \scaleobj{.8}{\int_{\BB_1}} \dd\mu(\sigma_1)\ \sigma_1^{\otimes n}\Big) &\leq \lim_{n\to\infty} \frac1n\,\rel{D}{\AA_n}{\co\big(\BB_1^{\otimes n,\, \mathrm{iid}}\big)} \\
&= \rel{D^\infty}{\AA}{\co\big(\BB_1^\mathrm{iid}\big)} \\
&= \liminf_{n\to\infty} \frac1n\,\rel{D}{\AA_n}{\co\big(\BB_1^{\otimes n,\, \mathrm{iid}}\big)} \\
&\leqt{(iii)} \inf_{\nu\in \PP(\BB_1)} \liminf_{n\to\infty} \frac1n\, D\Big(\AA_n\, \Big\|\, \scaleobj{.8}{\int_{\BB_1}} \dd\nu(\sigma_1)\ \sigma_1^{\otimes n}\Big) \\
&\leq \liminf_{n\to\infty} \frac1n\, D\Big(\AA_n\, \Big\|\, \scaleobj{.8}{\int_{\BB_1}} \dd\mu(\sigma_1)\ \sigma_1^{\otimes n}\Big)\, ,
\ee
where~(iii) is analogous to~\eqref{taking_measure_out_proof_eq0}. This is only possible if
\bb
\lim_{n\to\infty} \frac1n\, D\Big(\AA_n\, \Big\|\, \scaleobj{.8}{\int_{\BB_1}} \dd\mu(\sigma_1)\ \sigma_1^{\otimes n}\Big) = \rel{D^\infty}{\AA}{\co\big(\BB_1^\mathrm{iid}\big)}\, ,
\ee
where the limit exists. From~\eqref{taking_measure_out} we then obtain that
\bb
\rel{D^\infty}{\AA}{\co\big(\BB_1^\mathrm{iid}\big)} &= \min_{\nu\in \PP(\BB_1)} \liminf_{n\to\infty} \frac1n \inf_{\rho_n\in \AA_n} D\Big(\rho_n\, \Big\|\, \scaleobj{.8}{\int_{\BB_1}} \dd\nu(\sigma_1)\ \sigma_1^{\otimes n}\Big) \\
&\leq \minst_{\nu\in \PP(\BB_1)}\, \lim_{n\to\infty} \frac1n \inf_{\rho_n\in \AA_n} D\Big(\rho_n\, \Big\|\, \scaleobj{.8}{\int_{\BB_1}} \dd\nu(\sigma_1)\ \sigma_1^{\otimes n}\Big) \\
&\leq \lim_{n\to\infty} \frac1n \inf_{\rho_n\in \AA_n} D\Big(\rho_n\, \Big\|\, \scaleobj{.8}{\int_{\BB_1}} \dd\mu(\sigma_1)\ \sigma_1^{\otimes n}\Big) \\
&= \rel{D^\infty}{\AA}{\co\big(\BB_1^\mathrm{iid}\big)}\, ,
\ee
which proves also~\eqref{taking_measure_out_ordinary_limit}.
\end{proof}

\begin{rem}
The above proof of Lemma~\ref{taking_measure_out_lemma} is quite general, and it works also if instead of the relative entropy one were to consider a different quantum divergence $\mathds{D}(\cdot\|\cdot)$, with the only assumptions that it be: (a)~anti-monotonic in the second argument, and (b)~such that $\mathds{D}(\rho\|\lambda\sigma) = \mathds{D}(\rho\|\sigma) - \log \lambda$ for all pairs of states $\rho,\sigma$ and all $\lambda>0$. These assumptions are satisfied by most quantum divergences, including e.g.\ the max-relative entropy of~\cite{Datta08}.
\end{rem}

\subsection{Proof of Theorem~\ref{stronger_genq_Sanov_thm}} \label{subsec_proof_stronger_genq_Sanov}

Before we delve into the proof of Theorem~\ref{stronger_genq_Sanov_thm}, we need to introduce some terminology concerning measured relative entropies, which are indispensable tools to lift classical results to the quantum world. The \deff{measured relative entropy} between two quantum states $\rho$ and $\sigma$ on the same Hilbert space $\HH$ is defined as
\bb
D^{\all}(\rho\|\sigma) \coloneqq \sup_{\MM \in \all} \rel{D}{\MM(\rho)}{\MM(\sigma)}\, ,
\label{measured_relative_entropy}
\ee
where $\all$ denotes the set of all quantum measurements (POVMs) on the system, which we can think of as quantum-to-classical channels. In general, the measured relative entropy will be smaller than its quantum counterpart~\cite[Proposition~5]{Berta2017}. However, it is a fundamental fact of quantum mechanics that when the Hilbert is of the form $\HH^{\otimes n}$ and the second state is permutationally symmetric over the copies, the two are asymptotically very close. This key insight goes under the name of \emph{asymptotic spectral pinching}~\cite{Hiai1991, Hayashi2002, Sutter2017}. Here we report it in the form of~\cite[Lemma~2.4]{berta_composite}, with the explicit estimates in~\cite[Eq.~(6.16) and~(6.18)]{HAYASHI-GROUP}:

\begin{lemma}[{\cite[Lemma~2.4]{berta_composite}}] \label{pinching_lemma}
Let $\HH$ be a Hilbert space of dimension $d \coloneqq \dim(\HH) < \infty$, and let $\rho_n,\sigma_n \in \D\big(\HH^{\otimes n}\big)$ be two states over $n$ copies of the system. Assume that $\sigma_n$ is permutation invariant, in the sense that $U_\pi^{\vphantom{\dag}} \sigma_n U_\pi^\dag = \sigma_n$ for all permutations $\pi\in S_n$, where $U_\pi$ is the unitary that permutes the tensor factors of $\HH^{\otimes n}$ according to $\pi$. Then
\bb
D(\rho_n\|\sigma_n) - (d-1)\big(\tfrac{d}{2}+1\big) \log(n+1) \leq D^\all(\rho_n\|\sigma_n) \leq D(\rho_n\|\sigma_n)\, .
\ee
\end{lemma}

We are now ready to present the full proof of Theorem~\ref{stronger_genq_Sanov_thm}.

\begin{proof}[Proof of Theorem~\ref{stronger_genq_Sanov_thm}]
The first part of the argument is similar to that employed to prove~\cite[Theorem~14]{generalised-Sanov}, with the important difference that, instead of relying on~\cite[Theorem~8]{generalised-Sanov}, we employ the stronger~\cite[Theorem~4]{doubly-comp-classical}. Fix $k\in \N^+$, $\mathrm{b} \in \{\mathrm{iid},\mathrm{av}\}$, and write
\bb
\rel{D^{\all}}{\AA_k}{\co\big(\BB_1^{\otimes k,\,\mathrm{b}}\big)} &= \inf_{\rho_k\in \AA_k,\ \sigma_k \in \co\scaleobj{1.3}{(}\BB_1^{\otimes k,\,\mathrm{b}}\scaleobj{1.3}{)}}\ \sup_{\MM\in \all} \rel{D}{\MM(\rho_k)}{\MM(\sigma_k)} \\
&\eqt{(i)} \sup_{\MM\in \all}\ \inf_{\rho_k\in \AA_k,\ \sigma_k \in \co\scaleobj{1.3}{(}\BB_1^{\otimes k,\,\mathrm{b}}\scaleobj{1.3}{)}} \rel{D}{\MM(\rho_k)}{\MM(\sigma_k)}\, ,
\ee
where, as before, $\all$ denotes the set of all quantum measurements with finitely many outcomes, and the equality~(i) holds due to~\cite[Lemma~13]{brandao_adversarial}, because $\AA_k$ and $\co\big(\BB_1^{\otimes k,\,\mathrm{b}}\big)$ are both closed and convex,\footnote{The set $\co\big(\BB_1^{\otimes k,\,\mathrm{b}}\big)$ is closed because it is the convex hull of the compact set $\BB_1^{\otimes k,\,\mathrm{b}}$. The compactness of $\BB_1^{\otimes k,\,\mathrm{b}}$ follows from~\eqref{q_F_n_iid} and from the compactness of $\BB_1$.} and the set of all measurements is closed under `finitely labelled mixtures'. See also~\cite[Lemma~A.2]{berta_composite} for this special case. Due to the above equality, we can now fix a measurement $\MM_\star$ on $k$ copies of the system such that
\bb
\inf_{\rho_k\in \AA_k,\ \sigma_k \in \co\scaleobj{1.3}{(}\BB_1^{\otimes k,\,\mathrm{b}}\scaleobj{1.3}{)}} \rel{D}{\MM_\star(\rho_k)}{\MM_\star(\sigma_k)} \geq \rel{D^{\all}}{\AA_k}{\co\big(\BB_1^{\otimes k,\,\mathrm{b}}\big)} - 1\, .
\label{stronger_genq_Sanov_proof_eq2}
\ee

We can now devise the following strategy to perform hypothesis testing on $n$ copies of the system, for any positive integer $n$. We first divide the systems into $m \coloneqq \floor{n/k}$ batches comprising $k$ sub-systems each, discarding the rest. Note that as $n\in \N^+$ increases, $m$ takes all possible integer values. 

Due to the convexity of $\AA_n$ and to the fact that $\AA$ satisfies Axioms~\ref{q_ax_closed_permutations} and~\ref{q_ax_filtering}, discarding any $k$ sub-systems maps states in $\AA_{n+k}$ to states in $\AA_n$. To see this, it suffices to take $F_k=\id^{\otimes k}$ in~\eqref{complete_cone_preservation}; to trace away other sub-systems rather than the last $k$, simply apply a suitable permutation and exploit 
Axiom~\ref{q_ax_closed_permutations}. The same is true, rather more obviously, for the alternative hypotheses $\co\big(\BB_1^\mathrm{iid}\big)$ and $\co\big(\BB_1^\mathrm{av}\big)$.

Now, on each batch of $k$ sub-systems we apply the measurement $\MM_\star$, with outcome space $\XX$. We are thus left with a string of outcomes $x^m\in \XX^m$, which we treat as a random variable generated by an unknown probability distribution $P_m$. We then run a classical asymmetric hypothesis testing protocol between the following two hypotheses:
\begin{itemize}
\item[$\mathrm{H}_0$.] Null hypothesis: $P_m \in \RR_m$;
\item[$\mathrm{H}_1$.] Alternative hypothesis: $P_m \in \SS_n$.
\end{itemize}
Here, as $\RR_m$ and $\SS_m$ we choose the two sets of probability distributions
\bb
\RR_m \coloneqq \MM_\star^{\otimes m}\big(\AA_{mk} \big)
\label{q_to_c_reduction_sets_R_m}
\ee
and
\bb
\SS_m \coloneqq \left\{ \MM_\star\big(\sigma_k\big)^{\otimes m}:\ \sigma_k\in \co\big(\BB_1^{\otimes k,\, \mathrm{b}}\big)\right\} = \MM_\star\Big( \co\big(\BB_1^{\otimes k,\, \mathrm{b}}\big) \Big)^{\otimes m,\, \mathrm{iid}} ,
\label{q_to_c_reduction_sets_S_m}
\ee
where we employed the notation in~\eqref{q_F_n_iid} and, for a set of states $\FF_k \subseteq \D\big(\HH^{\otimes k}\big)$, we defined $\MM_\star(\FF_k) \coloneqq \left\{ \MM_\star(\sigma_k):\ \sigma_k\in \FF_k\right\}$. Doing so yields the inequality
\bb
\rel{\stein}{\AA}{\BB_1^{\mathrm{b}}} &\geq \frac1k\, \stein(\RR\|\SS)\, ,
\label{stronger_genq_Sanov_proof_eq5}
\ee
where the factor $1/k$ comes from the fact that we have consumed (asymptotically) $k$ quantum systems to produce each classical system. For more details on this relatively standard step, we refer the reader to the analogous quantum-to-classical reduction that leads to~\cite[Eq.~(39)]{berta_composite}.

To continue, we want to apply~\cite[Theorem~4]{doubly-comp-classical} to the classical setting. To this end, we verify assumptions~(a') and~(b) there:
\begin{enumerate}
\item[(a')] The fact that $\RR_1$ is closed and that any $\RR_m$ is convex follows immediately from the corresponding properties of $\AA_k$ and $\AA_{mk}$. Similarly, the fact that $\RR_m$ is closed under permutations of the symbols (i.e.~\cite[Axiom~III]{doubly-comp-classical} for $\RR$) descends directly from Axiom~\ref{q_ax_closed_permutations} for $\AA$. Verifying the closure under tensor powers of $\RR$~\cite[Axiom~II]{doubly-comp-classical} is also elementary: for all $m\in \N^+$ and $P = \MM_\star(\rho_k) \in \RR_1$, where $\rho_k\in \AA_k$, using Axiom~\ref{q_ax_tensor_powers} for $\AA$ we have
\bb
P^{\otimes m} = \MM_\star^{\otimes m}\big(\rho_k^{\otimes m}\big) \in \MM_\star^{\otimes m}(\AA_{mk}) = \RR_m\, .
\ee

As for~\cite[Axiom~I]{doubly-comp-classical}, which is nothing but the classical version of Axiom~\ref{q_ax_depolarising} in which we also set $k=1$, we use the state $\rho_1$ whose existence is guaranteed by Axiom~\ref{q_ax_depolarising} for $\AA$ to define $R \coloneqq \MM_\star\big(\rho_1^{\otimes k}\big)\in \RR_1$. By Axiom~\ref{q_ax_depolarising}, and due to the fact that we are in finite dimension, for all $\rho_{mk}\in \AA_{mk}$ there must exist a constant $C<\infty$ such that $\rho_{mk} \leq C \rho_1^{\otimes mk}$; then, by measuring we deduce that
\bb
\MM_\star^{\otimes m} (\rho_{mk}) \leq C\, \MM_\star^{\otimes m} \big(\rho_1^{\otimes mk}\big) = C\, \MM_\star\big(\rho_1^{\otimes k}\big)^{\otimes m} = C\, R^{\otimes m} ,
\ee
i.e.\ $\supp\big(\MM_\star^{\otimes m} (\rho_{mk})\big) \subseteq \supp(R)^m$. Since $\rho_{mk}\in \AA_{mk}$ was arbitrary, this shows directly that $\supp(P_m) \subseteq \supp(R)^m$ for all $P_m = \MM_\star^{\otimes m} (\rho_{mk}) \in \RR_m$. Also, for all $\delta\in [0,1]$ we see that
\bb
\mmv{R}^{\otimes m} \circ \MM_\star^{\otimes m} =  \left( \mmv{\,\MM_\star\left(\rho_1^{\otimes k}\right)} \circ \MM_\star\right)^{\otimes m} = \left(\MM_\star \circ \mmv{\,\rho_1^{\otimes k}}\right)^{\otimes m} = \MM_\star^{\otimes m} \circ \mmv{\,\rho_1^{\otimes k}}^{\otimes m} ,
\ee
where $\mmv{R} = \mmv{\,\MM_\star\left(\rho_1^{\otimes k}\right)}$ is the classical depolarising channel, defined by a formula identical to~\eqref{q_depolarising}. Applying this identity to an arbitrary $\rho_{mk}\in \AA_{mk}$ and using Axiom~\ref{q_ax_depolarising}(b) for $\AA$ shows immediately that $\RR$ satisfies~\cite[Axiom~I]{doubly-comp-classical}.

The only condition that remains to verify is~\cite[Axiom~V]{doubly-comp-classical}. Denoting by $U$ the uniform probability distribution on $\XX$, and fixing some $\lambda\in (0,1)$ to be determined later, define the classical channel $W = \mmq{\lambda}{U}: \XX\to \XX$ with transition probabilities $W(y|x) = (1-\lambda) \delta_{x,y} + \frac{\lambda}{|\XX|}$. Clearly, $W$ is informationally complete. 

To continue, we need to fix a notation for the POVM that represents the measurement $\MM_\star$, which will be denoted by $(E_x)_{x\in \XX}$. Remember that each $E_x$ is a positive semi-definite operator on $\HH$, and that $\sum_{x\in \XX} E_x = \id$. Now, for all $x_1,\ldots,x_{m-1},y_m\in \XX$, all probability distributions $Q_m = Q_{X_1\ldots X_m} \in \RR_m$, with $Q_m = \MM_\star^{\otimes m}(\rho_{mk})$ and $\rho_{mk}\in \AA_{mk}$, setting $Y_m \coloneqq W(X_m)$ we have
\bb
&\pr\{Y_m = y_m\}\, Q_{X_1\ldots X_{m-1}\,|\, Y_m = y_m}(x_1,\ldots, x_{m-1}) \\
&\qquad = Q_{X_1\ldots X_{m-1} Y_m}(x_1,\ldots,x_{m-1},y_m) \\
&\qquad = \sum_{x_m} W(y_m|x_m)\, Q_{X_1\ldots X_m}(x_1,\ldots,x_m) \\
&\qquad = \sum_{x_m} \left((1-\lambda)\delta_{x_m,y_m} + \tfrac{\lambda}{|\XX|}\right) \Tr \left[ \left(E_{x_1}\otimes \ldots \otimes E_{x_m}\right) \rho_{mk}\right] \\
&\qquad = \Tr \left[ \left(E_{x_1}\otimes \ldots \otimes E_{x_{m-1}} \otimes \left( (1-\lambda) E_{y_m} + \tfrac{\lambda}{|\XX|} \id \right) \right) \rho_{mk}\right] ;
\ee
in other words,
\bb
&\pr\{Y_m = y_m\}\, Q_{X_1\ldots X_{m-1}\,|\, Y_m = y_m} \\
&\quad = \MM_\star^{\otimes (m-1)}\left(\Tr_{(m-1)k+1,\ldots,\, mk} \left[ \left(\id^{\otimes (m-1)k} \otimes \left( (1-\lambda) E_{y_m} + \tfrac{\lambda}{|\XX|} \id \right) \right) \rho_{mk}\right] \right) ,
\ee
where the partial trace on the right-hand side is over the last $k$ sub-systems. Provided that $\lambda\in (0,1)$ is large enough, we will have $(1-\lambda) E_{y_m} + \tfrac{\lambda}{|\XX|} \id\in \pazocal{V}$ for all values of $y_m\in \XX$ simultaneously, where $\pazocal{V}$ is the neighbourhood from Axiom~\ref{q_ax_filtering} for $\AA$. This then ensures that $\pr\{Y_m = y_m\}\, Q_{X_1\ldots X_{m-1}\,|\, Y_m = y_m} \in \MM_\star^{\otimes (m-1)}\big( \cone(\AA_{(m-1)k})\big)$. Renormalising, this implies $Q_{X_1\ldots X_{m-1}\,|\, Y_m = y_m} \in \MM_\star^{\otimes (m-1)}\big( \AA_{(m-1)k}\big) = \RR_{m-1}$, completing the verification of~\cite[Axiom~V]{doubly-comp-classical}.

\item[(b)] $\SS_1 = \MM_\star\Big( \co\big(\BB_1^{\otimes k,\, \mathrm{b}}\big) \Big)$ is clearly convex, and hence it is star-shaped around any $R\in \SS_1$; choosing some $R$ in the relative interior of $\SS_1$, we also obtain that $\supp(Q)\subseteq \supp(R)$ for all $Q\in \SS_1$.
\end{enumerate}

We are now in position to continue from~\eqref{stronger_genq_Sanov_proof_eq5}, obtaining
\bb
\rel{\stein}{\AA}{\BB_1^{\mathrm{b}}} &\geq \frac1k\, \stein(\RR\|\SS) \\
&\eqt{(ii)} \frac1k\, D(\RR_1\|\SS_1) \\
&\eqt{(iii)} \frac1k \inf_{\substack{\rho_k \in \AA_k,\\ \sigma_k \in \co\scaleobj{1.3}{(}\BB_1^{\otimes k,\,\mathrm{b}}\scaleobj{1.3}{)}}} \rel{D}{\MM_\star(\rho_k)}{\MM_\star(\sigma_{k})} \\
&\geqt{(iv)} \frac1k\, \rel{D^\all}{\AA_k}{\co\big(\BB_1^{\otimes k,\,\mathrm{b}}\big)} - \frac{1}{k} \\
&\geqt{(v)} \frac1k\, \rel{D}{\AA_k}{\co\big(\BB_1^{\otimes k,\,\mathrm{b}}\big)} - (d-1)\big(\tfrac{d}{2}+1\big) \frac{\log(k+1)}{k} - \frac{1}{k}\, .
\label{stronger_genq_Sanov_proof_eq11}
\ee
Here: in~(ii) we applied~\cite[Theorem~4]{doubly-comp-classical}, which is possible because we just verified conditions~(a') and~(b) there; in applying it, we also remembered that $\SS_1$ is convex; in~(iii) we used~\eqref{q_to_c_reduction_sets_R_m} and~\eqref{q_to_c_reduction_sets_S_m}; (iv)~is simply~\eqref{stronger_genq_Sanov_proof_eq2}; finally, (v)~holds due to the asymptotic spectral pinching inequality (Lemma~\ref{pinching_lemma}). Let us provide a little more detail. 

Due to the convexity and closure under permutations of both $\AA_k$ and $\co\big(\BB_1^{\otimes k,\,\mathrm{b}}\big)$, Lemma~\ref{relent_perm_symm_lemma} guarantees that
\bb
\rel{D^\all}{\AA_k}{\co\big(\BB_1^{\otimes k,\,\mathrm{iid}}\big)} = \inf_{\substack{\rho_k \in \AA_k,\ \sigma_k \in \co\scaleobj{1.3}{(}\BB_1^{\otimes k,\,\mathrm{b}}\scaleobj{1.3}{)}, \\ \text{$\rho_k$, $\sigma_k$ perm.\ inv.}}} D^\all(\rho_k\|\sigma_{k})\, ,
\label{stronger_genq_Sanov_proof_eq12}
\ee
where on the right-hand side the infimum is further restricted to permutationally invariant $\rho_k$ and $\sigma_k$. With this step clarified, the inequality~(v) in~\eqref{stronger_genq_Sanov_proof_eq11} descends directly from Lemma~\ref{pinching_lemma}.

We can now take the limit supremum as $k\to\infty$ in~\eqref{stronger_genq_Sanov_proof_eq11}, obtaining that
\bb
\rel{\stein}{\AA}{\BB_1^{\mathrm{b}}} &\geq \limsup_{k\to\infty} \frac1k\, \rel{D}{\AA_k}{\co\big(\BB_1^{\otimes k,\,\mathrm{b}}\big)}\, .
\ee
On the other hand, Lemma~\ref{q_converse_double_Stein_lemma} guarantees that also the converse statement
\bb
\rel{\stein}{\AA}{\BB_1^{\mathrm{b}}} &\leq \rel{D^\infty}{\AA}{\co\big(\BB_1^{\mathrm{b}}\big)} = \liminf_{k\to\infty} \frac1k\, \rel{D}{\AA_k}{\co\big(\BB_1^{\otimes k,\,\mathrm{b}}\big)}
\ee
holds, implying that
\bb
\rel{\stein}{\AA}{\BB_1^{\mathrm{b}}} = \rel{D^\infty}{\AA}{\co\big(\BB_1^{\mathrm{b}}\big)} = \lim_{k\to\infty} \frac1k\, \rel{D}{\AA_k}{\co\big(\BB_1^{\otimes k,\,\mathrm{b}}\big)}\, ,
\label{stronger_genq_Sanov_proof_eq16}
\ee
where the limit exists. Taking $\mathrm{b} = \mathrm{iid}$ and using~\eqref{taking_measure_out_ordinary_limit} in Lemma~\ref{taking_measure_out_lemma} proves~\eqref{stronger_genq_Sanov_iid}. As for~\eqref{stronger_genq_Sanov_av}, we can write
\bb
\rel{\stein}{\AA}{\BB_1^{\mathrm{av}}}\ &\eqt{(vi)}\ \, \rel{\stein}{\AA}{\co\big(\BB_1^{\mathrm{av}}\big)} \\
&\eqt{(vii)}\ \, \rel{\stein}{\AA}{\co\big(\co(\BB_1)^{\mathrm{av}}\big)} \\
&\eqt{(viii)}\ \, \rel{\stein}{\AA}{\co(\BB_1)^{\mathrm{av}}} \\
&\eqt{(ix)}\ 
%\rel{D^\infty}{\AA}{\co\big(\co(\BB_1)^{\mathrm{av}}\big)} = 
\lim_{k\to\infty} \frac1k\, \rel{D}{\AA_k}{\co\big(\co(\BB_1)^{\otimes k,\,\mathrm{av}}\big)} \\
&\eqt{(x)}\ \lim_{k\to\infty} \frac1k\, \rel{D}{\AA_k}{\co\big(\co(\BB_1)^{\otimes k,\,\mathrm{iid}}\big)} \\
&\eqt{(xi)}\ \, \rel{\stein}{\AA}{\co(\BB_1)^\mathrm{iid}}\, .
\label{stronger_genq_Sanov_proof_eq17}
\ee 
The above derivation can be justified as follows. In~(vi) and~(viii) we used the already mentioned fact that the Stein exponent does not change if we take the convex hulls of the sets representing the hypotheses, as follows, for example, from the expression on the second line of~\eqref{q_Stein}. In~(vii) we noted the elementary identity $\co\big(\BB_1^{\otimes n,\, \mathrm{av}}\big) = \co\big(\co(\BB_1)^{\otimes n,\, \mathrm{av}}\big)$, which holds for any set $\BB_1$. The identity in~(ix) is an application of~\eqref{stronger_genq_Sanov_proof_eq16} with $\BB_1\mapsto \co(\BB_1)$ and $\mathrm{b}\mapsto \mathrm{av}$. Continuing, (x)~follows from the second claim in Proposition~\ref{av_to_iid_reduction_reg_relent_prop}, and in~(xi) we applied once again~\eqref{stronger_genq_Sanov_proof_eq16}, this time with $\BB_1\mapsto \co(\BB_1)$ and $\mathrm{b}\mapsto \mathrm{iid}$.

Finally, as before, the expression on the last line of~\eqref{stronger_genq_Sanov_av}, featuring the infimum over $\mu\in \PP(\BB_1)$ outside of the limit in $n$, follows from~\eqref{taking_measure_out_ordinary_limit} in Lemma~\ref{taking_measure_out_lemma}.
\end{proof}

\subsection{Some corollaries of Theorem~\ref{stronger_genq_Sanov_thm}} \label{subsec_proofs_corollaries}

As mentioned, two special cases of Theorem~\ref{stronger_genq_Sanov_thm} are of particular operational relevance in quantum information. First, there is the case where $\AA = \SEP = \big(\SEP_n\big)_n$ is the set of separable states on some bipartite system $AB$ with Hilbert space $\HH_{AB} = \HH_A\otimes \HH_B$; formally, at the $n$-copy level we set
\bb
\SEP_ n = \SEP_{A^n:B^n} \coloneqq \co\left\{ \rho_{A^n} \otimes \sigma_{B^n}:\ \rho_{A^n} \in \D\big(\HH_A^{\otimes n}\big)\, ,\ \sigma_{B^n}\in \D\big(\HH_B^{\otimes n}\big) \right\} ,
\label{n_copy_separable}
\ee
and then
\bb
\SEP \coloneqq \big(\SEP_n\big)_n\, .
\label{sequence_separable}
\ee
In this context, Theorem~\ref{stronger_genq_Sanov_thm} yields immediately Corollary~\ref{entanglement_testing_cor}, already reported in Section~\ref{sec_main_results} and proved below.

\begin{proof}[Proof of Corollary~\ref{entanglement_testing_cor}]
It suffices to argue that we can apply Theorem~\ref{stronger_genq_Sanov_thm} with $\AA \mapsto \SEP$ and $\BB_1\mapsto \FF_1$, with definitions as in~\eqref{n_copy_separable}--\eqref{sequence_separable}. To this end, we verify all the required properties of $\SEP$.

The fact that $\SEP_n$ is closed and convex for all $n$ is obvious by construction, once one notes that it is the convex hull of a compact set. Axiom~\ref{q_ax_tensor_powers} follows from the fact that the tensor product of separable states is separable. In Axiom~\ref{q_ax_depolarising}, we can take $\sigma_1 = \frac{\id_{AB}}{|AB|} \in \SEP_1$ as the maximally mixed state on $AB$, which is separable and has full support. Given some $\sigma_{nk} \in \SEP_{nk}$, we observe that $\mmv{\,\sigma_1^{\otimes k}}^{\otimes n}(\sigma_{nk})$ is a convex combination of states that are obtained by tracing out some $AB$ sub-systems of $\sigma_{nk}$ and replacing them with copies of $\sigma_1$. Both of these operations preserve separability, entailing that $\mmv{\,\sigma_1^{\otimes k}}^{\otimes n}(\sigma_{nk}) \in \SEP_{nk}$. Closure under permutations symmetry (Axiom~\ref{q_ax_closed_permutations}) is clear by inspection, directly from~\eqref{n_copy_separable}. 

The only assumption that remains to be checked is that $\SEP$ satisfies Axiom~\ref{q_ax_filtering}. Now, it is well known that $\id^{\otimes k}_{AB}$ is in the interior of the cone of separable operators (see e.g.~\cite{GurvitsBarnum}). If we choose a neighbourhood $\pazocal{V}$ of $\id^{\otimes k}_{AB}$ that is contained inside $\cone(\SEP_k)$, it is easy to verify directly that~\eqref{complete_cone_preservation} will hold. This is observed for the first time in this context in~\cite{Piani2009}, and discussed in detail also after Definition~4 in~\cite{brandao_adversarial}.
\end{proof}

Another important application of Theorem~\ref{stronger_genq_Sanov_thm} is to the resource theory of non-stabiliser states in quantum computation, a.k.a.\ quantum magic. In an $n$-qubit system with Hilbert space $(\C^2)^{\otimes n}$, the set of \deff{stabiliser states} can be defined as~\cite{Veitch2014}
\bb
\stab_n \coloneqq \co\left\{ U \ketbra{0^n} U^\dag:\ U\in \pazocal{C}_n \right\} ,
\label{stabiliser_states}
\ee
where $\ket{0^n}\coloneqq \ket{0}^{\otimes n}$ is the first state in the computational basis,\footnote{In fact, all computational basis states are equivalent for the purpose of the definition~\eqref{stabiliser_states}.} and $\pazocal{C}_n$ is the \deff{Clifford group} over $n$ qubits. In what follows, we consider a resource testing task, which we could call \deff{magic testing}, in which the null hypothesis is given by the sequence
\bb
\stab \coloneqq \big( \stab_{nm}\big)_m\, ,
\label{sequence_stabiliser_states}
\ee
where $n$ is a positive integer. We consider $n$ to be fixed, and omit the dependence of the sequence $\stab$ on it. See also~\cite{Hayashi-Sanov-2} for the application of the framework of quantum hypothesis testing to this setting.

\begin{cor} \label{magic_cor}
Let $n$ be a fixed positive integer, and let $\FF_1\subseteq \D\big((\C^2)^{\otimes n}\big)$ be a closed set of $n$-qubit states. The Stein exponents of the magic testing tasks with null hypothesis given by the set of stabiliser states and composite i.i.d.\ or arbitrarily varying alternative hypothesis, with base set $\FF_1$, can be expressed as
\bb
%\rel{\sanov}{\FF^{\mathrm{iid}}}{} = 
\rel{\stein}{\stab}{\FF_1^{\mathrm{iid}}} &= \rel{D^\infty}{\stab}{\co\big(\FF_1^{\mathrm{iid}}\big)} \\[.5ex]
&= \minst_{\mu\in \PP(\FF_1)}\, \lim_{m\to\infty} \frac1m \inf_{\sigma_{nm} \in\, \stab_{nm}} D\Big(\sigma_{nm} \,\Big\|\, \scaleobj{.8}{\int_{\FF_1}} \dd\mu(\rho_n)\ \rho_n^{\otimes m}\Big)
\ee
and
\bb
%\rel{\sanov}{\FF^{\mathrm{av}}}{} = 
\rel{\stein}{\stab}{\FF_1^{\mathrm{av}}} &= \rel{\stein}{\stab}{\co(\FF_1)^\mathrm{av}} \\
&= \rel{\stein}{\stab}{\co(\FF_1)^\mathrm{iid}} \\
&= \minst_{\mu\in \PP(\co(\FF_1))}\, \lim_{m\to\infty} \frac1m \inf_{\sigma_{nm} \in\, \stab_{nm}} D\Big(\sigma_{nm} \,\Big\|\, \scaleobj{.8}{\int_{\co(\FF_1)}} \dd\mu(\rho_n)\ \rho_{n}^{\otimes m}\Big)\, ,
\ee
respectively, where $\PP(\CC)$ denotes the set of probability measures on the compact set $\CC$, and $\minst_\mu$ indicates that we restrict the minimisation to those $\mu$ such that the inner limit in $m$ exists (such a set is non-empty).
\end{cor}

\begin{proof}
As before, we argue that we can apply Theorem~\ref{stronger_genq_Sanov_thm} with $\AA\mapsto \stab$ and $\BB_1\mapsto \FF_1$, with definitions as in~\eqref{stabiliser_states}--\eqref{sequence_stabiliser_states}. Convexity and closedness of $\stab_{nm}$ are again clear from~\eqref{stabiliser_states}. Axiom~\ref{q_ax_depolarising}--\ref{q_ax_closed_permutations} for $\stab$ can be verified with an argument that is entirely analogous to the one presented in the above proof of Corollary~\ref{entanglement_testing_cor}. For example, for Axiom~\ref{q_ax_depolarising} we can again choose $\sigma_1 = \id_{2^n}/2^n\in \stab_n$ as the maximally mixed state on $n$ qubits, which is a stabiliser state; the claim then follows precisely as before, because tracing out qubits and appending stabiliser states preserves the set of stabiliser states.

The only slightly delicate assumption, as usual, is Axiom~\ref{q_ax_filtering}. To verify it swiftly, the key step is to argue that $\id_{2}^{\otimes n}$ is in the interior of the cone $\cone(\stab_n)$ generated by stabiliser states. 
Since it is easy to verify that $\stab_n$ spans the whole space of Hermitian operators on $n$ qubits, this is equivalent to verifying that the maximally mixed state $\id_{2}^{\otimes n}/2^n$ is in the relative interior of $\stab_n$. Now, due to the fact that the Clifford group is a $1$-design, we have
\bb
\frac{1}{|\CC_n|} \sum_{U \in \CC_n} U\ketbra{0^n}U^\dag = \frac{\id_{2}^{\otimes n}}{2^n}\, .
\ee
The state on the right-hand side is thus the barycentre of the uniform measure on the set of pure stabiliser states, and as such it must belong to the relative interior of the polytope defined by their convex hull. This polytope, naturally, is nothing but $\stab_n$.

For each $n,m,k\in \N^+$, we can now take $\pazocal{V}$ as a neighbourhood of $\id_2^{\otimes nk}$ that is contained inside $\cone(\stab_{nk})$. If $F_{nk}\in \pazocal{V}$, therefore, we can find coefficients $\lambda(U)\geq 0$, where $U\in \CC_{nk}$, such that
\bb
F_{nk} = \sum_{U\in \CC_{nk}} \lambda(U)\, U\ketbra{0^{nk}} U^\dag .
\ee
For an arbitrary $\sigma_{n(m+k)}\in \stab_{n(m+k)}$, therefore,
\bb
&\Tr_{nm+1,\ldots,n(m+k)} \left[ \sigma_{n(m+k)} \big(\id_2^{\otimes nm} \otimes F_{nk} \big)\right] \\
&\qquad = \sum_{U\in \CC_{nk}} \lambda(U) \Tr_{nm+1,\ldots,n(m+k)} \left[ \sigma_{n(m+k)} \big(\id_2^{\otimes nm} \otimes U\ketbra{0^{nk}} U^\dag \big)\right] \\
&\qquad \in \cone\big(\stab_{nm}\big)\, ,
\ee
where the last line holds because the operations of applying a local Clifford unitary and measuring in the computational basis preserve the set of stabiliser states.
\end{proof}

\subsection{Composite i.i.d.\ or arbitrarily varying quantum hypotheses: proof of Theorem~\ref{q_both_composite_iid_or_av_thm}} \label{subsec_proofs_both_composite}

Theorem~\ref{q_both_composite_iid_or_av_thm} is not simply an application of Theorem~\ref{stronger_genq_Sanov_thm}, for it differs from this latter result in an important way. Namely, Eq.~\eqref{q_iid_vs_iid_convexity} features an optimisation over states $\rho\in \AA_1$ that is \emph{outside} of the regularisation, i.e.\ after the limit in $n$. As we argued in Section~\ref{sec_main_results}, on the one hand this makes the formula simpler; on the other, it does require a little more work. The following technical result, whose proof makes use of a version of the `Alicki--Fannes--Winter' trick from~\cite{Alicki-Fannes, tightuniform, Shirokov-review}, is key.

\begin{lemma} \label{extracting_rho_lemma}
Let $\AA_1\subseteq \D(\HH)$ be a non-empty closed set of states on a finite-dimensional Hilbert space $\HH$. Let $\BB = (\BB_n)_n$ be a sequence of sets of states $\BB_n\subseteq \D\big(\HH^{\otimes n}\big)$ that satisfies the following assumptions:
\begin{enumerate}[(a)]
\item for all $n\in \N^+$, the set $\BB_n$ is closed under partial trace of any single subsystem, in the sense that $\sigma_n\in \BB_n$ implies that $\Tr_k \sigma_n\in \BB_{n-1}$ for all $k\in\{1,\ldots,n\}$, where $\Tr_k$ denotes the partial trace over the $k^\text{th}$ sub-system; also,
\end{enumerate}
there exists some $\tau \in \BB_1$ such that:
\begin{enumerate}[(a)] \setcounter{enumi}{1}
\item $\supp(\sigma_n) \subseteq \supp(\tau)^{\otimes n}$ for all $\sigma_n\in \BB_n$; and
\item $\BB$ is closed under the insertion of $\tau$, in the sense that,
\bb
\tau^{[k]}\otimes \sigma_{n-1}^{[1,\ldots,k-1,\,k+1,\ldots n]}\in \BB_n\qquad \forall\ k\in\{1,\ldots,n\}\, ,\quad \forall\ \sigma_{n-1} \in \BB_{n-1}\, ,
\ee
where superscripts in square brackets denote the tensor factors where each state is acting.\footnote{For example, for $n=3$, $k=2$, and $\sigma_2 = \alpha \otimes \beta$, with this notation we have $\tau^{[2]} \otimes \sigma_2^{[1,3]} = \alpha \otimes \sigma \otimes \beta$. We can set by convention $\BB_0 = \{1\}$, where $1$ is the trivial state on the trivial system with Hilbert space $\C$.}
\end{enumerate}
Then
\bb
\liminf_{n\to\infty} \frac1n \inf_{\rho \in \AA_1} \rel{D}{\rho^{\otimes n}}{\BB_n} = \min_{\rho \in \AA_1} \liminf_{n\to\infty} \frac1n\, \rel{D}{\rho^{\otimes n}}{\BB_n}\, ,
\label{extracting_rho}
\ee
where the minimum on the right-hand side exists.
An analogous identity holds if one replaces $\liminf$ with $\limsup$ on both sides.
\end{lemma}

\begin{proof}
The fact that the left-hand side of~\eqref{extracting_rho} is no larger than the right-hand side is elementary, and follows by taking as an ansatz on the left-hand side a fixed $\rho\in \AA_1$. We now proceed to show the converse inequality. If there exists no $\rho\in \AA_1$ such that $\supp(\rho)\subseteq \supp(\tau)$ there is nothing to prove, as in that case, due to assumption~(b), we will necessarily have $\supp\big(\rho^{\otimes n}\big)\not\subseteq \supp(\sigma_n)$ for all $\rho\in \AA_1$ and all $\sigma_n \in \BB_n$, so that $\rel{D}{\rho^{\otimes n}}{\BB_n} = +\infty$, in turn implying that the left-hand side of~\eqref{extracting_rho} is infinite. Therefore, without loss of generality we can assume that 
\bb
\AA'_1 \coloneqq \AA_1 \cap \{\rho\in \D(\HH):\ \supp(\rho) \subseteq \supp(\tau)\}\neq \emptyset\, .
\label{extracting_rho_proof_eq0}
\ee

Let $I\subseteq \N^+$ be an infinite set such that
\bb
\liminf_{n\to\infty} \frac1n \inf_{\rho \in \AA_1} \rel{D}{\rho^{\otimes n}}{\BB_n} = \liminf_{n\to\infty} \frac1n \inf_{\rho \in \AA'_1} \rel{D}{\rho^{\otimes n}}{\BB_n}  = \lim_{n\in I} \frac1n \inf_{\rho \in \AA'_1} \rel{D}{\rho^{\otimes n}}{\BB_n}\, .
\label{extracting_rho_proof_eq1}
\ee
For all $n\in I$, let $\rho_n\in \AA'_1$ be a state with the property that
\bb
\rel{D}{\rho_n^{\otimes n}}{\BB_n} \leq \inf_{\rho \in \AA'_1} \rel{D}{\rho^{\otimes n}}{\BB_n} + 1\, .
\label{extracting_rho_proof_eq2}
\ee
By~\eqref{extracting_rho_proof_eq0}, it holds that $\supp(\rho_n) \subseteq \supp(\tau)$.
 
Since $\AA_1$ is a closed (and hence compact) set of states, the same is true of $\AA'_1$; we can therefore extract from the sequence $(\rho_n)_{n\in I}$ a subsequence $(\rho_n)_{n\in J}$, where $J\subseteq I$ is also infinite, such that $\rho_{n} \tends{}{n\in J} \rho$ for some $\rho\in \AA'_1$. Set 
\bb
\omega_n^{\pm} \coloneqq \frac{1}{\e_n}\left(\rho - \rho_n\right)_\pm\, ,\qquad \e_n \coloneqq \frac12 \left\|\rho - \rho_n \right\|_1 \ctends{}{n\in J}{-1pt} 0\, ,
\label{extracting_rho_proof_eq3}
\ee
where we denoted by $X_\pm \coloneqq \sum_i \max\{\pm x_i,0\}\, P_i$ the positive and negative parts of the Hermitian operator $X$ with spectral decomposition $X = \sum_i x_i P_i$. Note that $\supp\big(\omega_n^\pm \big) \subseteq \supp(\tau)$, so that, denoting by $c>0$ the minimal non-zero eigenvalue of $\tau$, we have
\bb
\omega_n^\pm \leq \tfrac1c\, \tau\, .
\label{extracting_rho_proof_eq5}
\ee

We can now proceed inspired by the proof of the second claim of~\cite[Corollary~8]{tightuniform}. Start by constructing the auxiliary function $g:[0,\infty) \to [0,\infty)$ defined by
\bb
g(x) \coloneqq (x+1) \log(x+1) - x\log x\, ,
\label{g}
\ee
where we can set $g(0)\coloneqq 0$ by continuity. For all $n\in J$, write
\bb
&\rel{D}{\rho^{\otimes n}}{\BB_n} - \rel{D}{\rho_n^{\otimes n}}{\BB_n} \\
&\quad = \sum_{k=0}^{n-1} \left(\rel{D}{\rho^{\otimes (n-k)} \otimes \rho_n^{\otimes k}}{\BB_n} - \rel{D}{\rho^{\otimes (n-k-1)} \otimes \rho_n^{\otimes (k+1)}}{\BB_n} \right) \\
&\quad \leqt{(i)} \sum_{k=0}^{n-1} \left(\e_n \left(\rel{D}{\rho^{\otimes (n-k-1)} \otimes \omega^+_n \otimes \rho_n^{\otimes k}}{\BB_n} - \rel{D}{\rho^{\otimes (n-k-1)} \otimes \omega^-_n \otimes \rho_n^{\otimes k}}{\BB_n}\right) + g(\e_n) \right) \\
&\quad = \e_n \sum_{k=0}^{n-1} \left(\rel{D}{\rho^{\otimes (n-k-1)} \otimes \omega^+_n \otimes \rho_n^{\otimes k}}{\BB_n} - \rel{D}{\rho^{\otimes (n-k-1)} \otimes \omega^-_n \otimes \rho_n^{\otimes k}}{\BB_n}\right) + n g(\e_n) \\
&\quad \leqt{(ii)} \e_n \sum_{k=0}^{n-1} \left(\rel{D}{\rho^{\otimes (n-k-1)} \otimes \omega^+_n \otimes \rho_n^{\otimes k}}{\BB_n} - \rel{D}{\rho^{\otimes (n-k-1)}\otimes \rho_n^{\otimes k}}{\BB_{n-1}}\right) + n g(\e_n) \\
&\quad \leqt{(iii)} \e_n \sum_{k=0}^{n-1} \left(\rel{D}{\rho^{\otimes (n-k-1)} \otimes \rho_n^{\otimes k}}{\BB_{n-1}} + \log \tfrac1c - \rel{D}{\rho^{\otimes (n-k-1)}\otimes \rho_n^{\otimes k}}{\BB_{n-1}}\right) + n g(\e_n) \\
&\quad = n \left(\e_n \log\tfrac1c + g(\e_n)\right) .
\label{extracting_rho_proof_eq6}
\ee
Here, (i)~is a direct application of the Alicki--Fannes--Winter method~\cite{Alicki-Fannes, tightuniform, Shirokov-review}, which guarantees that
\bb
D(\alpha\|\FF) - D(\beta\|\FF) \leq \e \left( D(\gamma_+\|\FF) - D(\gamma_-\|\FF) \right) + g(\e)
\ee
for all convex sets $\FF\subseteq \D(\HH)$ and all pairs of states $\alpha,\beta\in \D(\HH)$, where we set $\e\coloneqq \frac12 \left\|\alpha-\beta\right\|_1$ and $\gamma_\pm \coloneqq \frac1\e\left(\alpha-\beta\right)_\pm$, and the function $g$ is defined in~\eqref{g}. We also observed that
\bb
\left(\rho^{\otimes (n-k)} \otimes \rho_n^{\otimes k} - \rho^{\otimes (n-k-1)} \otimes \rho_n^{\otimes (k+1)}\right)_\pm &= \left( \rho^{\otimes (n-k-1)} \otimes (\rho - \rho_n) \otimes \rho_n^{\otimes (k+1)}\right)_\pm \\
&= \rho^{\otimes (n-k-1)} \otimes (\rho - \rho_n)_\pm \otimes \rho_n^{\otimes (k+1)} \\
&= \e_n\, \rho^{\otimes (n-k-1)} \otimes \omega_n^\pm \otimes \rho_n^{\otimes (k+1)} .
\ee
In~(ii), instead, we removed the $(n-k)^\text{th}$ sub-system from the second relative entropy term inside the sum; due to the data processing inequality and assumption~(a), said term cannot increase under this procedure. Inequality~(iii) makes use of assumption~(c) and of~\eqref{extracting_rho_proof_eq5}; to see how it is deduced, for ease of notation we look at the case $k=n-1$ (the other cases follow \emph{mutatis mutandis}):
\bb
\rel{D}{\omega^+_n \otimes \rho_n^{\otimes (n-1)}}{\BB_n} &\leq \rel{D}{\omega^+_n \otimes \rho_n^{\otimes (n-1)}}{\tau\otimes \BB_{n-1}} \\
&=  D(\omega_n^+ \|\tau) + \rel{D}{\rho_n^{\otimes (n-1)}}{\BB_{n-1}} \\
&\leq \log\tfrac1c + \rel{D}{\rho_n^{\otimes (n-1)}}{\BB_{n-1}}\, ,
\ee
where the first inequality descends from~(c), and the second from~\eqref{extracting_rho_proof_eq5} together with the operator monotonicity of the logarithm.

We are now ready to write the final chain of inequalities:
\bb
\inf_{\rho'\in \AA_1} \liminf_{n\to\infty} \frac1n\, \rel{D}{{\rho'}^{\,\otimes n}}{\BB_n}\, &\leq\ \liminf_{n\to\infty} \frac1n\, \rel{D}{\rho^{\,\otimes n}}{\BB_n} \\
&\leq\ \liminf_{n\in J} \frac1n\, \rel{D}{\rho^{\otimes n}}{\BB_n} \\
&\leqt{(iv)}\ \liminf_{n\in J} \left( \frac1n\, \rel{D}{\rho_n^{\otimes n}}{\BB_n} + \e_n \log\tfrac1c + g(\e_n) \right) \\
&\eqt{(v)}\ \liminf_{n\in J} \frac1n\, \rel{D}{\rho_n^{\otimes n}}{\BB_n} \\
&\leqt{(vi)}\ \liminf_{n\in J} \frac1n\left(\inf_{\rho'\in \AA'_1} \rel{D}{{\rho'}^{\,\otimes n}}{\BB_n} + 1 \right) \\
&=\ \liminf_{n\in J} \frac1n \inf_{\rho'\in \AA'_1} \rel{D}{{\rho'}^{\,\otimes n}}{\BB_n} \\
&\eqt{(vii)}\ \lim_{n\in I} \frac1n \inf_{\rho'\in \AA'_1} \rel{D}{{\rho'}^{\,\otimes n}}{\BB_n} \\
&\eqt{(viii)}\ \liminf_{n\to\infty} \frac1n \inf_{\rho' \in \AA_1} \rel{D}{\rho^{\otimes n}}{\BB_n}\, .
\label{extracting_rho_proof_eq9}
\ee
In~(iv) we employed~\eqref{extracting_rho_proof_eq6}, (v)~holds because $\e_n$ vanishes along $J$ due to~\eqref{extracting_rho_proof_eq3}, (vi)~follows from~\eqref{extracting_rho_proof_eq2}, in~(vii) we remembered that $J\subseteq I$, and, finally, in~(viii) we used~\eqref{extracting_rho_proof_eq1}. This completes the justification of the above chain of inequalities. 

Since we already argued that the leftmost side of~\eqref{extracting_rho_proof_eq9} cannot be strictly smaller than the rightmost side, the only possibility is that all inequalities in~\eqref{extracting_rho_proof_eq9} are, in fact, equalities. Furthermore, looking at the first line of~\eqref{extracting_rho_proof_eq9}, we realise that $\rho$ achieves the minimum on the right-hand side of~\eqref{extracting_rho}. This completes the proof.
\end{proof}

We are now ready to present the proof of our second main result, Theorem~\ref{q_both_composite_iid_or_av_thm}.

\begin{proof}[Proof of Theorem~\ref{q_both_composite_iid_or_av_thm}]
We start by dealing with the cases where the null hypothesis is arbitrarily varying. These are relatively straightforward, as the sequence of closed convex sets 
\bb
\AA = \co\big(\AA_1^\mathrm{av}\big) = \big(\co\big(\AA_1^{\otimes n,\, \mathrm{av}}\big)\big)_n 
\ee
turns out to satisfy the assumptions of Theorem~\ref{stronger_genq_Sanov_thm}. Indeed, Axioms~\ref{q_ax_tensor_powers} and~\ref{q_ax_closed_permutations} are easy to verify directly. As for Axiom~\ref{q_ax_depolarising}, one can pick $\rho_1\in \relint \big(\co(\AA_1)\big)$, so that $\supp(\omega)\subseteq \supp(\rho_1)$ for all $\omega\in \AA_1$, implying that $\supp(\omega_1\otimes \ldots \otimes \omega_n) \subseteq \supp(\rho_1)^{\otimes n}$ for all choices of $\omega_1,\ldots,\omega_n\in \AA_1$, and hence, by taking convex combinations, $\supp(\rho_n) \subseteq \supp(\rho_1)^{\otimes n}$ for all $\rho_n\in \co\big(\AA_1^{\otimes n,\, \mathrm{av}}\big)$. Verifying Axiom~\ref{q_ax_depolarising}(b) is elementary. Axiom~\ref{q_ax_filtering} is also, for once, immediate: it suffices to take as $\pazocal{V}$ the whole cone of positive semi-definite operators. Eq.~\eqref{q_av_vs_iid} and~\eqref{q_av_vs_av} then follow directly from from~\eqref{stronger_genq_Sanov_iid} and~\eqref{stronger_genq_Sanov_av}, respectively, once one remembers that $\rel{\stein}{\AA_1^\av}{\BB_1^\mathrm{b}} = \rel{\stein}{\co\big(\AA_1^\av\big)}{\BB_1^\mathrm{b}}$ for $\mathrm{b}\in \{\iid,\av\}$, due to~\eqref{convexify_Stein}.

Next, we move on to the cases where the null hypothesis is composite i.i.d. Here we can no longer apply Theorem~\ref{stronger_genq_Sanov_thm}, as Axiom~\ref{q_ax_depolarising} typically fails for $\AA_1^\iid$ or $\co\big(\AA_1^\iid\big)$. However, we can circumvent this obstacle by using the same ideas as in Theorem~\ref{stronger_genq_Sanov_thm} to reduce the task from quantum to classical; we will then effectively apply once again~\cite[Theorem~4]{doubly-comp-classical}, but this time going through condition~(a) rather than~(a'). We will accomplish this step more easily by applying~\cite[Corollary~24]{doubly-comp-classical} directly.

Take $k\in \N^+$, $\mathrm{b} \in \{\mathrm{iid},\mathrm{av}\}$, and write
\bb
\rel{D^{\all}}{\co\big(\AA_1^{\otimes k,\,\mathrm{iid}}\big)}{\co\big(\BB_1^{\otimes k,\,\mathrm{b}}\big)} &= \inf_{\substack{\rho_k\in \co\scaleobj{1.3}{(}\AA_1^{\otimes k,\,\mathrm{iid}}\scaleobj{1.3}{)},\\[.0pt] \sigma_k \in \co\scaleobj{1.3}{(}\BB_1^{\otimes k,\,\mathrm{b}}\scaleobj{1.3}{)}}}\ \sup_{\MM\in \all} \rel{D}{\MM(\rho_k)}{\MM(\sigma_k)} \\
&\eqt{(i)} \sup_{\MM\in \all}\ \inf_{\substack{\rho_k\in \co\scaleobj{1.3}{(}\AA_1^{\otimes k,\,\mathrm{iid}}\scaleobj{1.3}{)},\\[.0pt] \sigma_k \in \co\scaleobj{1.3}{(}\BB_1^{\otimes k,\,\mathrm{b}}\scaleobj{1.3}{)}}} \rel{D}{\MM(\rho_k)}{\MM(\sigma_k)}\, ,
\ee
where in~(i) we once again applied~\cite[Lemma~13]{brandao_adversarial} (or~\cite[Lemma~A.2]{berta_composite}). Now, pick a measurement $\MM_\star$ on $k$ copies of the system such that
\bb
\inf_{\substack{\rho_k\in \co\scaleobj{1.3}{(}\AA_1^{\otimes k,\,\mathrm{iid}}\scaleobj{1.3}{)},\\[.0pt] \sigma_k \in \co\scaleobj{1.3}{(}\BB_1^{\otimes k,\,\mathrm{b}}\scaleobj{1.3}{)}}} \rel{D}{\MM_\star(\rho_k)}{\MM_\star(\sigma_k)} \geq \rel{D^{\all}}{\co\big(\AA_1^{\otimes k,\,\mathrm{iid}}\big)}{\co\big(\BB_1^{\otimes k,\,\mathrm{b}}\big)} - 1\, .
\label{q_both_composite_iid_or_av_proof_eq3}
\ee

We can adopt the same hypothesis testing strategy as before, dividing up the $n$ available systems into batches of $k$ systems each (while discarding the rest), measuring each batch of $k$ copies with $\MM_\star$, and then applying a classical test on the string of outcomes. (Note that both hypotheses are closed under the operation of discarding a few sub-systems.) Doing so yields the bound
\begin{align}
\rel{\stein}{\AA_1^\mathrm{iid}}{\BB_1^\mathrm{b}} &\geqt{(ii)} \frac1k\, \Rel{\stein}{\big(\AA_1^{\otimes k,\,\mathrm{iid}}\big)^\mathrm{iid}}{\big(\BB_1^{\otimes k,\,\mathrm{b}}\big)^\mathrm{b}} \nonumber \\
&\geqt{(iii)} \frac1k\, \Rel{\stein}{\big(\co\big(\AA_1^{\otimes k,\,\mathrm{iid}}\big)\big)^\mathrm{iid}}{\big(\co\big(\BB_1^{\otimes k,\,\mathrm{b}}\big)\big)^\mathrm{b}} \nonumber \\
&\geqt{(iv)} \frac1k\, \Rel{\stein}{\MM_\star\big(\co\big(\AA_1^{\otimes k,\,\mathrm{iid}}\big)\big)^{\mathrm{iid}}}{\MM_\star\big(\co\big(\BB_1^{\otimes k,\,\mathrm{b}}\big)\big)^{\mathrm{b}}} \label{q_both_composite_iid_or_av_proof_eq4} \\
&\eqt{(v)} \frac1k\, \rel{D}{\MM_\star\big(\co\big(\AA_1^{\otimes k,\,\mathrm{iid}}\big)\big)}{\MM_\star\big(\co\big(\BB_1^{\otimes k,\,\mathrm{b}}\big)\big)} \nonumber \\
&\geqt{(vi)} \frac1k\, \rel{D^\all}{\co\big(\AA_1^{\otimes k,\,\mathrm{iid}}\big)}{\co\big(\BB_1^{\otimes k,\,\mathrm{b}}\big)} - \frac1k \nonumber \\
&\geqt{(vii)} \frac1k\, \rel{D}{\co\big(\AA_1^{\otimes k,\,\mathrm{iid}}\big)}{\co\big(\BB_1^{\otimes k,\,\mathrm{b}}\big)} - (d-1)\big(\tfrac{d}{2}+1\big) \frac{\log(k+1)}{k} - \frac1k\, . \nonumber
\end{align}
The above steps can be justified as follows. The inequality~(ii) holds because what we described is a possible strategy in quantum hypothesis testing. Note that $\big(\FF_1^{\otimes k,\,\mathrm{f}}\big)^{\otimes m,\,\mathrm{f}} = \FF_1^{\otimes mk,\,\mathrm{f}}$ for $\mathrm{f} \in\{ \mathrm{iid}, \mathrm{av}\}$ and any set $\FF_1$. In~(iii) we strengthened the task by enlarging the base sets --- clearly, by doing so the Stein exponent cannot increase. In~(vi) we applied the measurement $\MM_\star$ on every batch of sub-systems. Step~(v) is where our proof departs from that of Theorem~\ref{stronger_genq_Sanov_thm}, as we apply the classical result in~\cite[Corollary~24, Eq.~(203)]{doubly-comp-classical}. The inequality~(vi) follows from~\eqref{q_both_composite_iid_or_av_proof_eq3}, while that in~(vii) is again an application of the pinging inequality (Lemma~\ref{pinching_lemma}); as in the proof of Theorem~\ref{stronger_genq_Sanov_thm}, this is possible due to the fact that the infima over states in both arguments of the measured relative entropy can be restricted to permutationally symmetric states without loss of generality, due to Lemma~\ref{relent_perm_symm_lemma}. 

We can now take the limit superior of~\eqref{q_both_composite_iid_or_av_proof_eq4} as $k\to\infty$, obtaining that
\bb
\rel{\stein}{\AA_1^\mathrm{iid}}{\BB_1^\mathrm{b}} \geq \limsup_{k\to\infty} \frac1k\, \rel{D}{\co\big(\AA_1^{\otimes k,\,\mathrm{iid}}\big)}{\co\big(\BB_1^{\otimes k,\,\mathrm{b}}\big)}\, .
\ee
Since Lemma~\ref{q_converse_double_Stein_lemma} states that the converse inequality holds with the $\limsup$ replaced by the $\liminf$, we see that
\bb
\rel{\stein}{\AA_1^\mathrm{iid}}{\BB_1^\mathrm{b}} = \lim_{k\to\infty} \frac1k\, \rel{D}{\co\big(\AA_1^{\otimes k,\,\mathrm{iid}}\big)}{\co\big(\BB_1^{\otimes k,\,\mathrm{b}}\big)}\, .
\label{q_both_composite_iid_or_av_proof_eq6}
\ee

It now pays off to distinguish the two cases $\mathrm{b} = \mathrm{iid}$ and $\mathrm{b} = \mathrm{av}$. In the former case, 
\bb
\rel{\stein}{\AA_1^\mathrm{iid}}{\BB_1^\mathrm{iid}}\, &=\ \lim_{k\to\infty} \frac1k\, \rel{D}{\co\big(\AA_1^{\otimes k,\,\mathrm{iid}}\big)}{\co\big(\BB_1^{\otimes k,\,\mathrm{iid}}\big)} \\
&\eqt{(viii)}\ \lim_{k\to\infty} \frac1k\, \inf_{\rho\in \AA_1} \rel{D}{\rho^{\otimes k}}{\co\big(\BB_1^{\otimes k,\,\mathrm{iid}}\big)} \\
&\eqt{(ix)}\ \minst_{\mu\in \PP(\BB_1)}\, \lim_{k\to\infty} \frac1k\, \inf_{\rho\in \AA_1} \Rel{D}{\rho^{\otimes k}}{\scaleobj{.8}{\int_{\BB_1}} \dd\mu(\sigma)\ \sigma^{\otimes k}}\, ,
\label{q_both_composite_iid_or_av_proof_eq7}
\ee
where in~(viii) we employed~\cite[Lemma~2.5]{berta_composite} to remove the convex hull over i.i.d.\ states in the first argument, while in~(ix) we applied Lemma~\ref{taking_measure_out_lemma}. This proves~\eqref{q_iid_vs_iid}.

We now proceed with the proof of~\eqref{q_iid_vs_iid_convexity}, under the assumption that $\BB_1$ is closed and convex. Continuing from~\eqref{q_both_composite_iid_or_av_proof_eq6}, we write
\begin{align}
\rel{\stein}{\AA_1^\mathrm{iid}}{\BB_1^\mathrm{iid}}\, &=\ \lim_{k\to\infty} \frac1k\, \rel{D}{\co\big(\AA_1^{\otimes k,\,\mathrm{iid}}\big)}{\co\big(\BB_1^{\otimes k,\,\mathrm{iid}}\big)} \nonumber \\
&\eqt{(x)}\ \lim_{k\to\infty} \frac1k\, \rel{D}{\co\big(\AA_1^{\otimes k,\,\mathrm{iid}}\big)}{\co\big(\BB_1^{\otimes k,\,\mathrm{av}}\big)} \nonumber \\
&\eqt{(xi)}\ \lim_{k\to\infty} \frac1k\, \rel{D}{\AA_1^{\otimes k,\,\mathrm{iid}}}{\co\big(\BB_1^{\otimes k,\,\mathrm{av}}\big)} \nonumber \\
&=\ \lim_{k\to\infty} \frac1k\, \inf_{\rho\in \AA_1} \rel{D}{\rho^{\otimes k}}{\co\big(\BB_1^{\otimes k,\,\mathrm{av}}\big)} \nonumber \\
&=\ \liminf_{k\to\infty} \frac1k\, \inf_{\rho\in \AA_1} \rel{D}{\rho^{\otimes k}}{\co\big(\BB_1^{\otimes k,\,\mathrm{av}}\big)} \label{q_both_composite_iid_or_av_proof_eq8} \\
&\eqt{(xii)}\ \min_{\rho\in \AA_1}\, \liminf_{k\to\infty} \frac1k\, \rel{D}{\rho^{\otimes k}}{\co\big(\BB_1^{\otimes k,\,\mathrm{av}}\big)} \nonumber \\
&\eqt{(xiii)}\ \min_{\rho\in \AA_1}\, \lim_{k\to\infty} \frac1k\, \rel{D}{\rho^{\otimes k}}{\co\big(\BB_1^{\otimes k,\,\mathrm{av}}\big)} \nonumber \\
&\eqt{(xiv)}\ \min_{\rho\in \AA_1}\, \lim_{k\to\infty} \frac1k\, \rel{D}{\rho^{\otimes k}}{\co\big(\BB_1^{\otimes k,\,\mathrm{iid}}\big)} \nonumber \\
&\eqt{(xv)} \minst_{\substack{\rho\in \AA_1, \\[0pt] \mu\in \PP(\BB_1)}} \lim_{k\to\infty} \frac1k\, \Rel{D}{\rho^{\otimes k}}{\scaleobj{.8}{\int_{\BB_1}} \dd\mu(\sigma)\ \sigma^{\otimes k}}\, . \nonumber
\end{align}
Here, (x)~is an application of the second claim of Proposition~\ref{av_to_iid_reduction_reg_relent_prop}, while in~(xi) we used again~\cite[Lemma~2.5]{berta_composite}. The identity in~(xii) follows from Lemma~\ref{extracting_rho_lemma}, applied with $\BB_k\mapsto \co\big( \BB_1^{\otimes k,\,\mathrm{av}}\big)$. Note that a reasoning similar to that presented at the beginning of this proof shows that the condition in Lemma~\ref{extracting_rho_lemma}(b) is satisfied whenever we pick some $\tau\in \relint(\BB_1)$. Conditions Lemma~\ref{extracting_rho_lemma}(a) and Lemma~\ref{extracting_rho_lemma}(c) are also swiftly verified. It is worth remarking at this point that Lemma~\ref{extracting_rho_lemma} would not be applicable to sets of the form $\BB_k\mapsto \co\big( \BB_1^{\otimes k,\,\mathrm{iid}}\big)$, as condition~(c) would fail to hold: this is the reason why, in step~(x), we first replaced the composite i.i.d.\ alternative hypothesis with an arbitrarily varying one.
Continuing with the justification of~\eqref{q_both_composite_iid_or_av_proof_eq8}, in~(xiii) we observed that the limit in $k$ exists due to Fekete's lemma~\cite{Fekete1923}: indeed, as we already argued in even greater generality in~\eqref{av_to_iid_reduction_reg_relent_proof_eq20}--\eqref{av_to_iid_reduction_reg_relent_proof_eq21}, due to the fact that $\co\big(\BB_1^\av\big)$ is closed under tensor products, the sequence $k\mapsto \rel{D}{\rho^{\otimes k}}{\co\big(\BB_1^{\otimes k,\,\av}\big)}$ is sub-additive. Then, in~(xiv) we leveraged once again the second claim of Proposition~\ref{av_to_iid_reduction_reg_relent_prop}, this time with the choice $\AA_k \mapsto \big\{\rho^{\otimes k}\big\}$, for some fixed $\rho\in \AA_1$, while (xv)~follows from Lemma~\ref{taking_measure_out_lemma}. This completes the proof of~\eqref{q_iid_vs_iid_convexity}.

It remains to prove~\eqref{q_iid_vs_av}, which is now relatively straightforward: we can write
\bb
\rel{\stein}{\AA_1^\mathrm{iid}}{\co(\BB_1)^\mathrm{av}}\ \, &\eqt{(xvi)}\ \ \rel{\stein}{\AA_1^\mathrm{iid}}{\BB_1^\mathrm{av}} \\
&\eqt{(xvii)}\ \ \lim_{k\to\infty} \frac1k\, \rel{D}{\co\big(\AA_1^{\otimes k,\,\mathrm{iid}}\big)}{\co\big(\BB_1^{\otimes k,\,\mathrm{av}}\big)} \\
&\eqt{(xviii)}\ \ \lim_{k\to\infty} \frac1k\, \rel{D}{\co\big(\AA_1^{\otimes k,\,\mathrm{iid}}\big)}{\co\big(\co(\BB_1)^{\otimes k,\,\mathrm{av}}\big)} \\
&\eqt{(xix)}\ \ \rel{\stein}{\AA_1^\mathrm{iid}}{\co(\BB_1)^\mathrm{iid}}\, ,
\label{q_both_composite_iid_or_av_proof_eq9}
\ee
where~(xvi) can be justified as in the first three lines of~\eqref{stronger_genq_Sanov_proof_eq17}, (xvii)~holds by~\eqref{q_both_composite_iid_or_av_proof_eq6}, in~(xviii) we remembered that the two sets at the second argument are identical, as we already saw in step~(vii) of~\eqref{stronger_genq_Sanov_proof_eq17}, and in~(xix) we leveraged the equality between first and second line of~\eqref{q_both_composite_iid_or_av_proof_eq8}, with $\BB_1\mapsto \co(\BB_1)$. With this substitution, we can now apply~\eqref{q_both_composite_iid_or_av_proof_eq8} directly. Doing so allows us to obtain directly~\eqref{q_iid_vs_av} from~\eqref{q_both_composite_iid_or_av_proof_eq9}, thereby concluding the proof.
\end{proof}

\section{Conclusion} \label{sec_conclusion}

We have established a bouquet of new results in quantum hypothesis testing, and in particular obtained explicit (regularised) expressions for the Stein exponents corresponding to a general family of tasks in which the null hypothesis is subjected to very few assumptions, and in particular it is allowed to be composite and genuinely correlated, while the alternative hypothesis is more strongly constrained, and required to be either composite i.i.d.\ or arbitrarily varying. In doing so, we have extended and simplified both the `generalised quantum Sanov theorem'~\cite{generalised-Sanov} as well as prior results that covered only the case where both hypotheses are composite i.i.d.~\cite{berta_composite}. This generalisation comes, however, at a cost: while the expression obtained in~\cite{generalised-Sanov} is single letter, our formulas are not --- that is, they involve a regularisation over the number of copies. The results of~\cite{Mosonyi2022} seem to suggest that this feature is unavoidable, even in the most basic cases of composite alternative hypothesis.

It would be desirable to keep adding rows to Table~\ref{results_table}, solving the Stein exponent in ever more complex cases. While the ultimate goal would be to obtain a row with all green cells, the next natural step would be to look at a scenario where the null hypothesis is either composite i.i.d.\ or arbitrarily varying, while the alternative hypothesis is very general and possibly genuinely correlated. Solving this would further extend the validity of the generalised quantum Stein's lemma~\cite{Brandao2010, Hayashi-Stein, GQSL}, covering also the case of a composite (albeit not genuinely correlated) null hypothesis. The reader might wonder what our new techniques have to say in this context, given that this setting is superficially similar to ours --- it can be obtained from it by simply exchanging the two hypotheses. The answer is, rather disappointingly, that there is not much hope to apply our methods there, \emph{at least not directly}. The reason, in short, is that swapping the two hypotheses gives rise to a completely different problem. More in detail, our strategy is based on a quantum-to-classical reduction via measurements, and the experience with the generalised quantum Stein's lemma suggests that this is not a viable way to solve the Stein exponent when it is the alternative hypothesis the one that is genuinely correlated.

\smallskip
\noindent \emph{Acknowledgements.} I am grateful to Mario Berta for comments on a preliminary version of some of these results.  Funded by the European Union under the ERC StG ETQO, Grant Agreement no.~101165230.

\bibliography{../../biblio}

\end{document}